\documentclass[10pt]{article}
\pdfoutput=1

\usepackage{caption}
\usepackage{subcaption}

\usepackage[backend=bibtex, style = numeric, doi = false, url = false, isbn = false, maxbibnames = 6]{biblatex}
\bibliography{jointBib}     
\renewbibmacro{in:}{}      
\newbibmacro{string+doi}[1]{\iffieldundef{doi}{#1}{\href{https://dx.doi.org/\thefield{doi}}{#1}}}
\DeclareFieldFormat{title}{\usebibmacro{string+doi}{\mkbibemph{#1}}}
\DeclareFieldFormat[article]{title}{\usebibmacro{string+doi}{\mkbibquote{#1}}}
\DeclareFieldFormat[incollection]{title}{\usebibmacro{string+doi}{\mkbibquote{#1}}}                   
\DeclareFieldFormat[inproceedings]{title}{\usebibmacro{string+doi}{\mkbibquote{#1}}}

\usepackage[margin=1em]{caption}
\usepackage[utf8]{inputenc}

\usepackage{amssymb}
\usepackage{amsfonts}
\usepackage{amsmath}
\usepackage{graphicx}
\usepackage{indentfirst}
\usepackage{dsfont}

\allowdisplaybreaks[4]

\usepackage{mathtools}

\usepackage{enumerate}
\usepackage{url}
\usepackage{lmodern}
\usepackage{newpxtext}

\usepackage{array}
\usepackage{multirow}
\usepackage{longtable}

\usepackage{amsthm}
\newtheorem{remark}{Remark}
\newtheorem{theorem}{Theorem}
\newtheorem{corollary}[theorem]{Corollary}
\newtheorem{proposition}[theorem]{Proposition}
\newtheorem{lemma}[theorem]{Lemma}

\usepackage[symbol]{footmisc}

\usepackage{tikz,pgfplots}
\definecolor{darkblue}{rgb}{0.15,0.35,0.55}
\definecolor{reddish}{rgb}{.8, 0.2, 0.2}
\definecolor{plotblue}{RGB}{0,119,187}
\definecolor{plotgreen}{RGB}{0,153,136}
\definecolor{plotorange}{RGB}{238,119,51}
\definecolor{plotmagenta}{RGB}{238,51,119}
\definecolor{plotgray}{RGB}{128,128,128}
\definecolor{plotcyan}{RGB}{51,187,238}
\definecolor{plotred}{RGB}{204,51,17}
\usepgfplotslibrary{fillbetween}

\usepackage[pdfpagelabels, linktocpage=true,pdfusetitle]{hyperref}
\hypersetup{
	colorlinks=true,
	citecolor=darkblue,
	linkcolor=reddish,
	urlcolor=darkblue,
	pdfauthor={},
	pdfsubject={},
	breaklinks=true,
}

\usepackage{cleveref}

\long\def\ca#1\cb{} 

\newcommand{\becs}{\begin{cases}}
	\newcommand{\bem}{\begin{matrix}}

		\newcommand{\dya}[1]{|#1\rangle\langle#1|}

		\newcommand{\encs}{\end{cases}}
	\newcommand{\enm}{\end{matrix}}

\newcommand{\inpd}[2]{\langle#1|#2\rangle }
\newcommand{\ket}[1]{|#1\rangle }

\newcommand{\ot}{\otimes }

\newcommand{\Tr}{{\rm Tr}}

\newcommand{\BC}{{\mathcal B}}
\newcommand{\CC}{{\mathcal C}}
\newcommand{\DC}{{\mathcal D}}

\newcommand{\FC}{{\mathcal F}}

\newcommand{\HC}{{\mathcal H}}
\newcommand{\IC}{{\mathcal I}}

\newcommand{\LC}{{\mathcal L}}
\newcommand{\MC}{{\mathcal M}}
\newcommand{\NC}{{\mathcal N}}
\newcommand{\OC}{{\mathcal O}}
\newcommand{\PC}{{\mathcal P}}
\newcommand{\QC}{{\mathcal Q}}
\newcommand{\RC}{{\mathcal R}}

\newcommand{\TC}{{\mathcal T}}

\newcommand{\WC}{{\mathcal W}}

\newcommand{\bB}{\textbf{b}}

\newcommand{\pB}{\textbf{p}}
\newcommand{\qB}{\textbf{q}}

\newcommand{\sB}{\textbf{s}}
\newcommand{\tB}{\textbf{t}}

\newcommand{\vB}{\textbf{v}}
\newcommand{\wB}{\textbf{w}}
\newcommand{\xB}{\textbf{x}}
\newcommand{\yB}{\textbf{y}}

\newcommand{\Cbb}{\mathbb{C}}

\newcommand{\al}{\alpha }
\newcommand{\bt}{\beta }
\newcommand{\gm}{\gamma }

\newcommand{\dl}{\delta }
\newcommand{\Dl}{\Delta }
\newcommand{\ep}{\epsilon}

\newcommand{\zt}{\zeta }
\renewcommand{\th}{\theta } 

\newcommand{\Th}{\Theta }

\newcommand{\Lm}{\Lambda }

\newcommand{\sg}{\sigma }

\newcommand{\om}{\omega }

\DeclareMathOperator{\tr}{tr}
\newcommand{\one}{\mathds{1}}
\newcommand{\herm}{^{\mathsf{H}}}
\DeclareMathOperator{\Herm}{Herm}

\DeclareMathOperator{\Span}{span}

\usepackage[normalem]{ulem}

\def\outl#1{\par{\medskip\noindent\hspace*{0.1cm}\bf
		\mathversion{bold}#1\mathversion{normal}\smallskip} }
   \def\xa{} \def\xb{}  

\def\outl#1{}\def\xa{}\def\xb{}

\ca
\def\outl#1{\par{\medskip\noindent\hspace*{.5cm}\bf
		\mathversion{bold}#1\mathversion{normal}\smallskip} }
\long\def\xa#1\xb{} %
 
\cb

  \usepackage[affil-it]{authblk}

\title{The platypus of the quantum channel zoo}
\author[1,2,3]{Felix Leditzky\thanks{\href{mailto:leditzky@illinois.edu}{\texttt{leditzky@illinois.edu}}
	}
}
\author[2,3]{Debbie Leung\thanks{\href{mailto:wcleung@uwaterloo.ca}{\texttt{wcleung@uwaterloo.ca}}}
}
\author[4,5]{Vikesh Siddhu\thanks{
    \href{mailto:vsiddhu@protonmail.com}{\texttt{vsiddhu@protonmail.com}}, Present address:
IBM Quantum}}

\author[4,6]{Graeme Smith\thanks{\href{mailto:graeme.smith@colorado.edu}{graeme.smith@colorado.edu}}}
\author[7]{John A. Smolin\thanks{\href{mailto:smolin@us.ibm.com}{smolin@us.ibm.com}}}

\affil[1]{Department of Mathematics \& IQUIST, University of Illinois at Urbana-Champaign}
\affil[2]{Institute for Quantum Computing \& Department of Combinatorics and Optimization, University of Waterloo}
\affil[3]{Perimeter Institute for Theoretical Physics}
\affil[4]{JILA, University of Colorado Boulder}
\affil[5]{Department of Physics \& Quantum Computing Group, Carnegie Mellon University}
\affil[6]{CTQM \& Department of Physics, University of Colorado Boulder}
\affil[7]{IBM Quantum, IBM T.J.~Watson Research Center}

\begin{document}
	\maketitle

\begin{abstract}
    Understanding quantum channels and the strange behavior of their capacities
    is a key objective of quantum information theory.  Here we study a
    remarkably simple, low-dimensional, single-parameter family of quantum
    channels with exotic quantum information-theoretic features.  As the
    simplest example from this family, we focus on a qutrit-to-qutrit channel
    that is intuitively obtained by hybridizing together a simple degradable
    channel and a completely useless qubit channel. Such hybridizing makes this
    channel's capacities behave in a variety of interesting ways. For instance,
    the private and classical capacity of this channel coincide and can be
    explicitly calculated, even though the channel does not belong to any class
    for which the underlying information quantities are known to be additive.
    Moreover, the quantum capacity of the channel can be computed explicitly,
    given a clear and compelling conjecture is true.  This ``spin alignment
    conjecture,'' which may be of independent interest, is proved in certain
    special cases and additional numerical evidence for its validity is
    provided.  Finally, we generalize the qutrit channel in two ways, and the
    resulting channels and their capacities display similarly rich behavior.
    In the companion paper \cite{paper2PRL}, we further show that the qutrit
    channel demonstrates superadditivity when transmitting quantum information
    jointly with a variety of assisting channels, in a manner unknown before. 
\end{abstract}

\newpage
\tableofcontents
\newpage

\section{Introduction}\label{sec:introduction}

Quantum channels model noisy communication links between quantum parties.  The
channel noise affecting signals can be mitigated by encoding the messages
across many channel uses.  The highest rate at which information can be sent
reliably is known as a capacity of the channel.  Depending on the type of
information to be transmitted, we obtain different capacity quantities; for
example, a quantum channel may be used to transmit classical, quantum, or
private classical information.  The capacities in each case are the classical
capacity $\CC$, the quantum capacity $\QC$, and the private classical capacity
$\PC$, measured in bits per channel use, qubits per channel use, and private
bits per channel use, respectively.  The various capacities of a quantum
channel quantify its usefulness in the respective communication setting.

There is a variety of synergy effects that may occur during a quantum channel's
transmission of the various types of (quantum, private, classical) information.
These include super-additivity of coherent information
\cite{ShorJohn96,DiVincenzoShorEA98}, private information
\cite{SmithRenesEA08}, Holevo information \cite{Hastings09}, superactivation of
quantum capacity \cite{SmithYard08}, and private communication at a rate above
the quantum
capacity~\cite{HorodeckiHorodeckiEA00,HorodeckiHorodeckiEA05,LeungLiEA14} (see
Section~\ref{sec:prelim} for details).  These nonadditive effects enable
exciting and novel communication protocols, but at the same time they obscure a
succinct mathematical characterization of the corresponding quantum channel
capacities.

Remarkably, such nonadditivities appear to be common; some are exhibited even
for simple channels such as the depolarizing channel.  On the other hand, for
certain classes of quantum channels the information quantities listed above can
be additive, thus simplifying their information-theoretical characterization.
Unfortunately this is only known to be true in a few special cases.  For
example, the coherent information of a PPT channel is additive and indeed zero
\cite{HorodeckiHorodeckiEA00}.  The only other channel classes with known
additive coherent information are the degradable \cite{DevetakShor05} and
antidegradable channels, which as a result have a so-called single-letter
formula for their quantum capacity.  These channels also have the pleasing
property of having their private capacity equal to the quantum capacity
\cite{Smith08}; as a result, also the private capacity is given by a
single-letter formula for these channels.  There exists a smattering of
channels whose Holevo information is additive, and therefore these channels
have a single-letter formula for their classical
capacity~\cite{King02,King03,Shor02,KingMatsumotoEA07}.  However, beyond
special examples of proven nonadditivities and additivities, little is known
about most capacities of most channels.

The best path towards a deeper understanding of nonadditivity effects in
quantum information---in fact, a better understanding of quantum information
itself---is to better understand and develop the menagerie of these phenomena.
However, clean and clear examples of channels that isolate different aspects of
nonadditivity are in short supply.  As a result, over the past two decades
significant effort has been dedicated to elucidating these phenomena, leading
to numerous exciting findings; yet a full understanding still remains elusive.
Without such understanding, we lack a theory on how to best communicate with
quantum channels, and fail to answer the kinds of questions resolved in
classical information theory.  This is substantiated by the fact that random
codes can be suboptimal; that we cannot evaluate capacities beyond special
examples; that the known capacity quantities may not fully capture the
communication potential of a noisy channel; and that our understanding of error
correction in the quantum setting is incomplete, whether the data is classical,
private, or quantum.

This paper studies novel examples of additivity based on combining two
well-understood but very different classes of channels.  
In the simplest case, we combine two qubit-channels in a qutrit channel so that their inputs overlap along one dimension.  
The superposed identities of the subchannels used in the
transmission convey additional information to the receiver and thus becomes an integral
part of the communication sent. 

Our channel is constructed as follows.  We start with a degradable channel,
which we know has single-letter quantum and private capacities.  We then
explicitly break degradability by adding an extra input state that lets the
sender transmit quantum information directly to the environment but allows no
additional information to be sent to the output.  This makes it impossible for
the channel output to simulate the environment as required for degradability,
and thus takes the channel outside any class known to have additive private
or coherent information.  Nevertheless, the coherent information appears to
remain additive, which we can show up to a very reasonable conjecture.  Even
more surprisingly, the private information remains additive but takes a
\emph{much} larger value than the original channel.  This difference in private
and quantum capacity is a clear signature of the nondegradability of the
channel.  The apparent additivity of coherent information must therefore be
coming from some new mechanism of additivity, which we seek to understand.  

\subsection{Main results}


We study a remarkably simple, single-parameter family of qutrit channels $\NC_s$, 
along with generalizations $\MC_d$ and $\OC$ of this channel family to arbitrary dimension $d$. 
These channel families exhibit many strange behaviors for quantum communication while having uncomplicated classical and private classical capacities:  
The classical and private classical capacities can be calculated explicitly because the underlying information quantities (Holevo and private information, respectively) are additive, even though neither channel family belongs to any of the known additivity classes.  
The same holds true for the quantum capacity: the coherent informations of these  families are additive, provided that a certain entropy minimization conjecture is true.  
We give evidence for the validity of this ``spin alignment conjecture'' in the main text.
Despite the additivity of the capacities, both channel families have strictly larger private capacity than quantum capacity.
The simplicity of our channel families enables us to further generalize the study to a 3-parameter family of channels displaying rich behavior, including separations between all three capacities, which we analyze numerically.

The results in the current paper are even more surprising in the light of
additional findings in a companion paper \cite{paper2PRL}:
The coherent information of the channel $\NC_s$ tensored with an assisting
channel is super-additive, for a large swath of values of $s$ and for some
generically chosen assisting channel.  
The super-additivity can be lifted to quantum capacity for degradable assisting
channels if the spin alignment conjecture holds.  The assisting
channel can have positive or vanishing quantum capacity.  The mechanism behind
this superadditivity is novel and in particular differs from the known
explanation of super-activation \cite{SmithYard08,Oppenheim08}.
For the $d$-dimensional generalization $\MC_d$ of $\NC_{1/2}$, we focus on
super-additivity with a $(d{-}1)$-dimensional erasure channel and recover
results similar to those for $\NC_s$.  In addition, even stronger qualitative
results can be obtained.  First, super-additivity of quantum capacity can be
proved unconditionally (without the spin alignment conjecture).  Second, the
effect holds for all nontrivial values of the erasure probability, for an
appropriately large local dimension $d$.

\subsection{Structure of the paper}

This paper is structured as follows.  In Section~\ref{sec:prelim} we give some
general background on quantum channels, their various capacities, and special
classes of channels.  The main objects of this paper, the channel $\NC_s$ and
its $d$-dimensional generalizations $\MC_d$ and $\OC$, are defined in
Section~\ref{sec:channel-defs}.  The channel coherent information and the
quantum capacity of these channels are discussed in Sections~\ref{sec:CI} and
\ref{sec:qcap}, respectively.  Section~\ref{sec:spinAl} formulates the spin
alignment conjecture relative to which the coherent information of $\NC_s$ is
additive.  We then derive bounds on the quantum capacity of $\NC_s$ and $\MC_d$
in Section~\ref{sec:qcap-bounds}, and determine their private and classical
capacities exactly in Section~\ref{sec:pccap}.
Section~\ref{sec:cap-discussion} summarizes and further discusses the results
on the capacities of $\NC_s$.  Finally, we present the 3-parameter
generalization of $\NC_s$ and a numerical analysis of its capacities in
Section~\ref{sec:3parameter}.
MATLAB and Python code used to obtain the numerical results mentioned above will be made available at \cite{github}.

\section{Preliminaries}
\label{sec:prelim}
\xb 
\outl{Hilbert space, dual, space of linear operators, isometry, channel pair,
Unital, coherent information, channel coherent information, quantum channel
capacity} 
\xa

\subsection{Quantum channels}\label{sec:capacities}

Let $\HC$ be a Hilbert space of finite dimension $d$.  Let $\HC^{\dag}$ be the
dual of $\HC$, and $\hat \HC \cong \HC \ot \HC^{\dag}$ be the space of linear
operators acting on $\HC$.  Let $\HC_a, \HC_b,$ and $\HC_c$ be three Hilbert
spaces of dimensions $d_a, d_b,$ and $d_c$ respectively.  An isometry $E:\HC_a
\mapsto \HC_b \ot \HC_c$, i.e., a map satisfying $E^{\dag} E = I_a$~(the
identity on $\HC_a$), takes an input Hilbert space $\HC_a$ to a subspace of a
pair of output spaces $\HC_b \ot \HC_c$.  This isometry generates a quantum
channel pair, $(\BC, \BC^c)$, i.e., a pair of completely positive trace
preserving~(CPTP) maps, with superoperators,
\begin{align}
    \BC(X) = \Tr_c(EXE^{\dag}), \quad \text{and} \quad
    \BC^c(X) = \Tr_b(EXE^{\dag}),
    \label{eq:chanPair}
\end{align}
that take any element $X \in \hat \HC_a$ to $\hat \HC_b$ and $\hat \HC_c$,
respectively.  Each channel in this pair $(\BC,\BC^c)$ may be called the
complement of the other. 
The isometry $E$ can be written as $E = \sum_i K_i \ot \ket{i}_c =
\sum_j L_j \ot \ket{j}_b$ where $K_i: \HC_a \mapsto \HC_b$ and $L_j : \HC_a
\mapsto \HC_c$ are Kraus operators for the $\BC$ and $\BC^c$ channels
respectively, satisfying $\BC(X_a) = \sum_{i=1}^{d_c} K_i X_a K_i^\dagger$ and $\sum_{i=1}^{d_c} K_i^\dagger K_i = I_a$, and similarly for $\lbrace L_j\rbrace_{j=1}^{d_b}$.

If the input of the isometry $E$ is restricted to a
subspace $\HC_{\bar a}$ of $\HC_a$, then such a restricted map is still an
isometry on $\HC_{\bar a}$ and defines a pair of channels $(\bar \BC, \bar
\BC^c)$, where each channel $\bar \BC$ and $\bar \BC^c$ is called a {\em
sub-channel} of $\BC$ and $\BC^c$, respectively.  When focussing on some quantum
channel $\BC$, it is common to denote $\HC_a,\HC_b,$ and $\HC_c$ as the channel
input, output, and environment respectively.  Any CPTP map~(together with its
complement) may be written as~\eqref{eq:chanPair} in terms of a suitable
isometry $E$.  
Another representation of a CPTP map comes from its Choi-Jamio\l{}kowski
operator.  To define this operator, consider a linear map $\BC \colon \hat
\HC_a \mapsto \hat \HC_b$ and a maximally entangled state,
\begin{align}
    \ket{\phi} = \frac{1}{\sqrt{d_a}} \sum_i \ket{i}_{a} \ot \ket{i}_{a},
    \label{eq:maxEnt}
\end{align}
on $\HC_{a} \ot \HC_a$. The {\em unnormalized} Choi-Jamio\l{}kowski operator
of $\BC$ is
\begin{align}
    J^{\BC}_{ab} = d_a(\IC_{a} \ot \BC)  (\dya{\phi}),
    \label{eq:cjOp}
\end{align}
where $\IC_a$ denotes the identity map acting on $\hat \HC_a$.  The linear map
$\BC$ is CPTP if and only if the above operator is positive semidefinite and
its partial trace over $\HC_b$ is the identity $I_{a}$ on $\HC_{a}$.

\subsection{Quantum capacity}

The quantum capacity $\QC(\BC)$ of a quantum channel
$\BC\colon\hat{\HC_a}\to\hat{\HC_b}$ is defined as the largest rate at which
quantum information can be sent faithfully through the channel.  It can be
expressed in terms of an entropic quantity as follows.  Let $\rho_a$ denote a
density operator~(unit trace positive semi-definite operator) on $\HC_a$ and
for any $\rho_a$ let $\rho_b \coloneqq \BC(\rho_a)$ and $\rho_c \coloneqq
\BC^c(\rho_a)$.  The coherent information~(or entropy bias) of a channel $\BC$
at a density operator $\rho_a$ is
\begin{align}
    \Dl(\BC, \rho_a) = S(\rho_b) - S(\rho_c),
    \label{entBias}
\end{align}
where $S(\rho) = -\Tr(\rho \log \rho)$~(we use $\log$ base $2$ by default) is
the von-Neumann entropy of $\rho$. The channel coherent information~(sometimes
called the {\em single-letter} coherent information),
\begin{align}
    \QC^{(1)}(\BC) = \max_{\rho_a} \Dl(\BC, \rho_a),
    \label{eq:q1Deg}
\end{align}
represents an achievable rate for sending quantum information across the
channel $\BC$, and hence $\QC(\BC)\geq
\QC^{(1)}(\BC)$~\cite{Lloyd97,Shor02a,Devetak05}. The maximum achievable rate
is equal to the quantum capacity of $\BC$, and given by a {\em multi-letter}
formula~(sometimes called a regularized
expression)~\cite{BarnumNielsenEA98,BarnumKnillEA00,Lloyd97,Shor02a,Devetak05},
\begin{align}
    \QC(\BC) = \sup_{n \in \mathbb{N}} \frac{1}{n} \QC^{(1)}(\BC^{\ot n}),
    \label{eq:chanCap}
\end{align}
where $\BC^{\ot n}$ represent $n \in \mathbb{N}$ parallel~(sometimes called
joint) uses of $\BC$.  The regularization in~\eqref{eq:chanCap} is necessary because the
channel coherent information is {\em super-additive}, i.e., for any two quantum
channels $\BC$ and $\BC'$ used together, the channel coherent information of
the joint channel $\BC \ot \BC'$ satisfies an inequality,
\begin{align}
    \QC^{(1)}(\BC \ot \BC') \geq \QC^{(1)}(\BC) + \QC^{(1)}(\BC'),
    \label{eq:nonAddCI}
\end{align}
which can be strict~\cite{DiVincenzoShorEA98, FernWhaley08,SmithSmolin07,
SmithYard08, LeditzkyLeungEA18, BauschLeditzky19, BauschLeditzky20,
SiddhuGriffiths20, Siddhu20, Siddhu21}. 
The coherent information $\QC^{(1)}(\BC)$ is said to be {\em weakly additive}
if equality holds in~\eqref{eq:nonAddCI} whenever $\BC'$ is a tensor power of
$\BC$. If this equality holds for arbitrary $\BC'$ then $\QC^{(1)}(\BC)$ is
said to be {\em strongly additive}.

\subsection{Private capacity}

The private capacity $\PC(\BC)$ of a quantum channel
$\BC\colon\hat{\HC_a}\to\hat{\HC_b}$ is operationally defined as the largest
rate at which classical information can be faithfully sent through the channel
in such a way that the environment, $\HC_c$, gains no meaningful knowledge
about the information being sent.  A formula for the private capacity was
derived by Cai et al.~\cite{CaiWinterEA04,Devetak05} in terms of a quantity called the
channel private information:
\begin{align}
    \PC^{(1)}(\BC) = \max_{\lbrace p_x,\rho_a^x\rbrace} \left[ \Dl(\BC,\bar{\rho}_a) - \sum\nolimits_x p_x \Dl(\BC,\rho_a^x) \right].
    \label{eq:private-information}
\end{align}
Here, $\Dl(\BC,\rho)$ denotes the coherent information of the channel $\BC$
with respect to the state $\rho$ as defined in \eqref{entBias}, the
maximization is over all quantum state ensembles $\lbrace
p_x,\rho_a^x\rbrace_x$, and $\bar{\rho}_a = \sum_x p_x \rho_a^x$ denotes the
ensemble average of the density operators $\{\rho_a^x\}$ over the probability
distribution $\{p_x\}$~($p_x>0$ and $\sum_x p_x = 1$),
Restricting the maximization in~\eqref{eq:private-information} to an ensemble
of pure states $\rho_a^x$ makes $\Dl(\BC, \rho_a^x)=0$,
reducing~\eqref{eq:private-information} to \eqref{eq:q1Deg}, and resulting in a
maximum value which simply equals the channel coherent information
$\QC^{(1)}(\BC)$. This value is at most $\PC^{(1)}(\BC)$, i.e., 
\begin{align}
   \QC^{(1)}(\BC) \leq \PC^{(1)}(\BC).
   \label{eq:q1LeqP1}
\end{align}
There are channels for which this inequality is strict,
thus for such channels $\PC^{(1)}(\BC)$ cannot be obtained using an ensemble of
pure states alone.

A channel's private information can also be written as 
\begin{align}
    \PC^{(1)}(\BC) = \max_{\lbrace p_x,\rho_a^x\rbrace} \left[ I(\mathsf{x};b)_\sigma - I(\mathsf{x};c)_\sigma\right],
    \label{eq:pi-2}
\end{align}
where the mutual information at a quantum state $\sg$, 
$I(a;b)_{\sg} \coloneqq S(\sg_a) + S(\sg_b) - S(\sg_{ab})$, is evaluated above at
\begin{align} 
    \sigma_{\mathsf{x}be} = (\one_{\mathsf{x}}\otimes J) 
    \left(\sum\nolimits_x p_x [x]_{\mathsf{x}}\otimes \rho_a^x\right) 
    (\one_{\mathsf{x}}\otimes J)^\dagger,
    \label{eq:sigma-xbe}
\end{align}
with the label $\mathsf{x}$ denoting a classical register, $J$ denoting the
$\BC$ channel isometry~\eqref{eq:chanPair}, and the notation
\begin{align}
    [\psi] = \dya{\psi},
    \label{eq:notation}
\end{align}
represents a dyad onto any ket $\ket{\psi}$.
Cai et al.~\cite{CaiWinterEA04} and Devetak \cite{Devetak05} proved that the
private information $\PC^{(1)}(\BC)$ is an achievable rate for private
information transmission, $\PC(\BC)\geq \PC^{(1)}(\BC)$, and furthermore the
private capacity is bounded from above by the regularized private information.
As a result, we have the following coding theorem for the private capacity
\cite{CaiWinterEA04,Devetak05}:
\begin{align}
    \PC(\BC) = \sup_{n\in\mathbb{N}} \frac{1}{n} \PC^{(1)}(\BC^{\otimes n}).
	\label{eq:Pcap}
\end{align}
Except for a few special cases such as degradable
channels~\cite{DevetakShor05,Smith08}~(for definition see
Sec.~\ref{sec:spChanCl}) the regularization in~\eqref{eq:Pcap} is required. 
This requirement arises from super-additivity of $\PC^{(1)}$
~\cite{SmithRenesEA08}, i.e., $\PC^{(1)}$ satisfies an inequality of the form
in~\eqref{eq:nonAddCI}.
From~\eqref{eq:q1LeqP1} it follows that each term, $\QC^{(1)}(\BC^{\ot n})$, in the
limit~\eqref{eq:chanCap} is at most $\PC^{(1)}(\BC^{\ot n})$, and thus a channel's
quantum capacity $\QC(\BC)$ is at most its private classical capacity
$\PC(\BC)$~\eqref{eq:Pcap},
\begin{align}
    \QC(\BC) \leq \PC(\BC).
	\label{eq:QC-PC-relation}
\end{align}

\subsection{Classical capacity}

The classical capacity $\CC(\BC)$ of a quantum channel
$\BC\colon\hat{\HC_a}\to\hat{\HC_b}$ is defined as the largest rate at which
classical information can be faithfully sent through the channel.  In contrast
to private information transmission discussed above, there is no security
criterion involving the environment, and hence 
\begin{align}
    \PC(\BC) \leq \CC(\BC).
	\label{eq:PC-relation}
\end{align}
The classical capacity can be expressed in terms of an
information quantity called the Holevo information $\chi(\BC)$:
\begin{align}
    \chi(\BC) = \max_{\lbrace p_x,\rho_a^x\rbrace} \left[ 
    S(\bar{\rho}_b) - \sum\nolimits_x p_x S(\rho_b^x)
    \right] = \max_{\lbrace p_x,\rho_a^x\rbrace} I(\mathsf{x};b)_\sigma,
    \label{eq:chi}
\end{align}
where $\bar{\rho}_a = \sum_x p_x \rho_a^x$ and $\sigma_{\mathsf{x}be}$ is
defined in terms of the quantum state ensemble $\lbrace p_x,\rho_a^x\rbrace_x$
as in \eqref{eq:sigma-xbe} above. Quantum states in the
maximization~\eqref{eq:chi} of $\chi(\BC)$ can be chosen to be
pure~\cite{Wilde16}.
The classical capacity $\CC(\BC)$ of a quantum channel $\BC$ can be expressed
as~\cite{Holevo98,SchumacherWestmoreland97},
\begin{align}
    \CC(\BC) = \sup_{n\in\mathbb{N}} \frac{1}{n} \chi(\BC^{\otimes n}).
    \label{eq:Ccap}
\end{align}
Again, apart from a few special classes of channels the regularization of the
Holevo information in \eqref{eq:Ccap} is necessary due to super-additivity of
$\chi(\BC)$~\cite{Hastings09,Shor04}.
Comparing~\eqref{eq:chi} with~\eqref{eq:pi-2} reveals $\PC^{(1)}(\BC) \leq
\chi^{(1)}(\PC)$, as a result, each term, $\PC^{(1)}(\BC^{\ot n})$, in the
limit~\eqref{eq:Pcap} is at most $\chi^{(1)}(\PC^{\ot n})$, and thus a channel's
private channel capacity $\PC(\BC)$ is at most its classical capacity
$\CC(\BC)$~\eqref{eq:PC-relation}.

\subsection{Entanglement-assisted classical capacity}

Finally, we discuss the entanglement-assisted classical capacity $\CC_{E}(\BC)$
of a quantum channel $\BC$, which is defined as the optimal rate of faithful
classical information transmission when the sender and receiver have access to
unlimited entanglement assisting the encoding and decoding process.  It can be
expressed in terms of the channel mutual information $I(\BC)$, defined as
\begin{align}
    I(\BC) = \max_{\rho_a} I(a';b)_\sigma = \max_{\rho_{a}} \left[S(\sg_{a'}) + S(\sigma_b) - S(\sigma_{a'b}) \right],
    \label{eq:entasscap} 
\end{align}
with $\sigma_{a'b} = (\IC_{a'}\otimes \BC)(\psi_{a'a})$ and
$|\psi\rangle_{a'a}$ an (arbitrary) purification of the input state $\rho_a$.
Note that $I(a';b)$ is concave in the input state $\rho_a$ and can therefore be
computed efficiently~\cite{Wilde16,FawziFawzi18}.  The entanglement-assisted
classical capacity of a quantum channel is equal to its channel mutual
information~\cite{BennettShorEA02}: 
\begin{align}
    \CC_{E}(\BC) = I(\BC) \,,
\end{align}
and is an upper bound of $\CC(\BC)$ by definition.

\subsection{Special channel classes}
\label{sec:spChanCl}
\xb 
\outl{Unital, Non-additivity, Less noisy, more capable, additivity, Degradable,
anti-degradable, additivity. PPT, Hadamard, Entanglement Breaking}
\xa

If a channel maps the identity element at its input to the identity at its
output, then the channel is called a {\em unital channel}. 
A channel $\BC$ is called {\em degradable}, and its complement $\BC^c$ {\em
anti-degradable}, if there is another channel $\DC$ such that $\DC \circ \BC =
\BC^c$~\cite{DevetakShor05,CubittRuskaiEA08}. Sometimes this channel $\DC$ is
called the {\em degrading map} of the degradable channel $\BC$. For any two
channels $\BC'$ and $\BC$, each either degradable or anti-degradable, the joint
channel $\BC \ot \BC'$ has additive coherent information, i.e., equality holds
in~\eqref{eq:nonAddCI}~\cite{DevetakShor05, LeditzkyDattaEA18}.  For a
degradable channel $\BC$, the coherent information $\Dl(\BC, \rho_a)$ is
concave in $\rho_a$~\cite{YardHaydenEA08}, and thus $\QC^{(1)}(\BC)$ can be
computed with relative ease~\cite{FawziFawzi18, RamakrishnanItenEA21}. As a
result the quantum capacity of a degradable channel, which simply equals
$\QC^{(1)}(\BC)$, can also be computed efficiently.  An anti-degradable channel has no
quantum capacity due to the no-cloning theorem.  An instance of an anti-degradable channel is a {\em
measure-and-prepare} or {\em entanglement-breaking}~(EB)
channel~\cite{HorodeckiShorEA03}.  An EB channel is one whose
Choi-Jamio\l{}kowski operator~\eqref{eq:cjOp} is
separable~\cite{HorodeckiShorEA03}.  The complement of an EB channel is called
a {\em Hadamard channel} \cite{Wilde16}.
Besides anti-degradable channels, the only other known class of zero-quantum-capacity
channels are {\em entanglement binding} or {\em positive under
partial-transpose}~(PPT) channels~\cite{HorodeckiHorodeckiEA00}.  A channel is
PPT if its Choi-Jamio\l{}kowski operator~\eqref{eq:cjOp} is positive under
partial transpose.  If a channel $\BC^c$ has zero quantum capacity then its
complement $\BC$ is called a {\em more capable} channel. A more capable channel
has equal quantum and private capacity. If a more capable channel $\BC$ has a
complement $\BC^c$ with zero private capacity, then $\BC$ has additive coherent
information in the sense that equality holds in eq.~\eqref{eq:nonAddCI} with
$\BC'=\BC$~\cite{Watanabe12}.

\section{Definition of channels}\label{sec:channel-defs}
\subsection{The $\NC_s$ channel}
\label{sec:NsChan}
\xb
\outl{Channel isometry $F$, $d_a = d_b = 3$, $d_c = 2$, Schmidt coefficient
$s \leq 1-s$, Channel pair $\NC_s, \NC_s^c$, Two key properties: degradable
sub-channels $\NC_{si}$ and perfect sub-channel to environment~(proofs at end
of Sec.), intuitive picture, outside the class of (anti)~degradable~(proof
follows from property), Unital~(simple cal), not EB/PPT~(proof $\NC_s$ later,
$\QC^{(1)}>0$), Hadamard~(EB argument)}
\xa

Let $\HC_a, \HC_b,$ and $\HC_c$ have dimensions $d_a = d_b = 3$, and $d_c=2$.
Consider an isometry $F_s\colon \HC_a \mapsto \HC_b \ot \HC_c$ of the form,
\begin{align}
    F_s \ket{0} &= \sqrt{s} \; \ket{0} \ot \ket{0} + \sqrt{1-s} \; \ket{1} \ot \ket{1},  \nonumber \\
    F_s \ket{1} &= \ket{2} \ot \ket{0}, \nonumber \\
    F_s \ket{2} &= \ket{2} \ot \ket{1},
    \label{isoDef1}
\end{align}
where $0 \leq s \leq 1/2$. This isometry was introduced previously by one of us
in~\cite{Siddhu21} with $\ket{1}$ and $\ket{2}$ in $\HC_a$ exchanged.
Its Kraus operators are unitarily equivalent to those of a channel $\LC_\alpha$ introduced
even earlier in~\cite{WangDuan18} and studied further in~\cite{WangXieEA17}.
The isometry can be written as $F_s = \sum_i K_i \ot \ket{i}$ where Kraus
operators $K_i: \HC_a \mapsto \HC_b$ match those in Sec.~IV.C
of~\cite{WangXieEA17} if one permutes the computational basis of $\HC_a$ as $\ket{0}_a \mapsto
\ket{2}_a \mapsto \ket{1}_a \mapsto \ket{0}_a$, exchanges $\ket{1}_b$ and
$\ket{2}_b$, and rewrites $s$ as $\sin ^2 \al$, $0 \leq \al \leq \pi/4$.
Through an equation of the form~\eqref{eq:chanPair} the
isometry~\eqref{isoDef1} gives rise to a complementary pair of channels
$\NC_s\colon \hat \HC_a \mapsto \hat \HC_b$ and $\NC^c_s \colon \hat \HC_a
\mapsto \hat \HC_c$. This channel pair has two simple properties.  

The first property, proved in Sec.~\ref{sec:proofsChan}, is the existence of
degradable sub-channels of $\NC^c_s$, obtained by restricting the channel input
to operators on a qubit subspace,
\begin{align}
    \HC_{ai} = \rm span \{ \ket{0}, \ket{i} \},
    \label{aiSub}
\end{align}
where $i$ is either $1$ or $2$. This restriction also results in an
anti-degradable sub-channel of $\NC_s^c$.
Quantum states lying solely in this qubit input sub-space can be used to send
an equal amount of quantum and private information to the $\NC_s$ output
$\HC_b$ but such states cannot be used to send any quantum or private
information to the $\NC_s^c$ output $\HC_c$.

The second channel property, also proved in Sec.~\ref{sec:proofsChan}, is the presence
of a perfect sub-channel of $\NC_s^c$ obtained by restricting the channel input
to operators on a qubit subspace $\HC_{a'}$ spanned by $\{ \ket{1}, \ket{2}
\}$. 
Quantum states lying solely in this qubit input subspace $\HC_{a'}$ can be used
to perfectly send information to the $\NC_s^c$ output $\HC_c$ while sending no
information to the $\NC_s$ output $\HC_b$. This perfect transmission of a qubit
to $\HC_c$ implies that the quantum, private, and classical capacities of
$\NC_s^c$ are at least $1$.  Since the dimension of the channel output $\HC_c$
is two, all these capacities of $\NC_s^c$ are at most $1$, thus
\begin{align}
    \QC(\NC_s^c)  = \PC(\NC_s^c) = \CC(\NC_s^c) = 1.
    \label{eq:NsCompCap}
\end{align}

Together, the intuitive picture above along with the two simple properties of
$\NC_s$ and $\NC_s^c$ help classify each of these channels.  For instance, one
can easily infer that each channel in the $(\NC_s, \NC_s^c)$ pair is neither
degradable nor anti-degradable.
If $\NC_s$ was degradable, then all its sub-channels would also be degradable.
However the sub-channel of $\NC_s$ obtained by restricting its input to
$\HC_{a'}$ is anti-degradable. This anti-degradable sub-channel is the
complement of a perfect~(and hence degradable) sub-channel to $\HC_c$ obtained
by restricting the channel input to $\HC_{a'}$. In similar vein, $\NC_s^c$ is
not degradable. Consider a sub-channel of $\NC_s^c$ obtained by restricting its
input to $\HC_{ai}$.  This sub-channel is anti-degradable since its complement,
a sub-channel of $\NC_s$ with input $\HC_{ai}$, is degradable.

Since each channel in the $(\NC_s, \NC_s^c)$ pair is neither degradable nor
anti-degradable, both channels are not EB since EB channels are anti-degradable.
A Hadamard channel is the complement of an EB channel, and since each channel
in the complementary pair $(\NC_s, \NC_s^c)$ is not EB, both the channels are not
Hadamard channels.
Another class of interesting channel are more capable channels. As mentioned in
Sec.~\ref{sec:spChanCl}, a more capable channel has a complement with no
quantum capacity. Both $\NC_s$ and $\NC_s^c$ don't belong to this more capable
channel class because these complementary channels both have non-zero quantum
capacity. An argument showing that the quantum capacity of $\NC_s$ is non-zero
was presented previously. An argument showing that the quantum capacity of
$\NC_s$ is non-zero is given in Sec.~\ref{sec:CI} and also in~\cite{Siddhu21}.
Since both $\NC_s$ and $\NC_s^c$ have non-zero quantum capacity, these channels
are not PPT which have zero quantum capacity.  Finally, it is easy to verify
that each channel in the $(\NC_s, \NC_s^c)$ pair is not unital.

\newcommand{\chA}{\RC_1}
\newcommand{\chB}{\RC_2^s}
\newcommand{\isomA}{R_1}
\newcommand{\isomB}{R_2^s}

The $\NC_s$ channel can also be viewed as a hybrid of two simple qubit input channels.
The first channel, $\chA$, perfectly maps its  input, $\HC_1$, to
the environment $\HC_c$. The second channel, $\chB$ for $s\in[0,1/2]$, is a degradable channel.
To define the first channel, $\chA\colon \hat \HC_1 \mapsto \hat \HC_b$, 
we use $\{\ket{1}, \ket{2}\}$ to label an orthonormal basis of $\HC_1$.
Then an isometry $\isomA\colon \HC_1 \mapsto \HC_b \ot \HC_c$ of the form
\begin{align}
    \begin{aligned}
        \isomA \ket{1} = \ket{2} \ot \ket{0}, \quad \text{and} \quad
        \isomA \ket{2} = \ket{2} \ot \ket{1},
    \end{aligned}
    \label{eq:F2Iso}
\end{align}
defines $\chA(X) = \Tr_{c}(\isomA X \isomA^{\dag}) = \Tr(X) \dya{2}$. Clearly
$\chA$ traces out its input to a fixed pure state while sending everything to
the environment. Thus $\chA$ cannot send any information to its output,
$\QC(\chA) = \PC(\chA) = \CC(\chA) = 0$.
To define the second channel $\chB$, we use $\{ \ket{0}, \ket{2} \}$ to label
an orthonormal basis of $\HC_2$.
Then $\chB\colon \hat \HC_{2} \mapsto \hat
\HC_{b}$ is generated by an isometry, $\isomB\colon \HC_{2} \mapsto \HC_{b} \ot
\HC_{c}$, of the form
\begin{align}
    \begin{aligned}
        \isomB \ket{0} = \sqrt{s} \; \ket{0} \ot \ket{0} + \sqrt{1-s} \; \ket{1} \ot \ket{1}, 
        \quad \text{and} \quad
        \isomB \ket{2} = \ket{2} \ot \ket{1},
    \end{aligned}
    \label{eq:F1Iso}
\end{align}
where $0 \leq s \leq 1/2$.
One can show that $\chB$ is degradable and thus $\QC^{(1)}(\chB) = \QC(\chB) =
\PC(\chB)$. The channel's classical capacity is one bit. This can be shown by
noticing that a bit encoded into inputs $\ket{0}$ and $\ket{2}$ leads to
orthogonal outputs achieving a rate of $1$, which saturates the maximal entropy
determined by the dimensional bound of the input of $\chB$. 

Notice $\HC_a$ is the union of $\HC_1$ and $\HC_2$ and the isometry $F_s :
\HC_a \mapsto \HC_b \ot \HC_c$, which gives rise to $\NC_s$, can be written as
a hybrid,
\begin{align}
    \begin{aligned}
        F_s \ket{0} &= \isomB \ket{0}, \quad
        F_s \ket{1} &= \isomA \ket{1}, \quad \text{and} \quad
        F_s \ket{2} &= \isomA \ket{2} = \isomB \ket{2},
    \end{aligned}
    \label{eq:hybrid}
\end{align}
of isometries $\isomA$ and $\isomB$ that give rise to $\chA$ and $\chB$, respectively.

\subsection{The $\MC_d$ channel}
\label{sec:MdChan}
\xb
\outl{Channel isometry $G$, $d_a = d_b = d$, $d_c = d-1$, equal Schmidt
coefficient $\frac{1}{\sqrt{d-1}}$, Channel pair $\MC_d, \MC_d^c$, repeat two
key properties: degradable sub-channels $\MC_{di}$ and perfect sub-channel to
environment~(proofs at end of Sec.), intuitive picture, outside the class of
(anti)~degradable~(proof follows from property), Unital~(simple cal), 
EB/PPT~(proof $\NC_s$ later, $\QC^{(1)}>0$), and Hadamard~(EB argument) channels}
\xa

We now consider a higher dimensional generalization of the isometry
in~\eqref{isoDef1} with $s=1/2$. This generalization, $G:\HC_a
\mapsto \HC_b \ot \HC_c$ with $G^{\dag} G = I_a$, operates on Hilbert spaces with dimensions $d_a = d_b
= d$, $d_c = d-1$, and $G$ has the form,
\begin{align}
    G \ket{0} = \frac{1}{\sqrt{d-1}}\sum_{j=0}^{d-2} \ket{j} \ot \ket{j}, \quad 
    G \ket{i} = \ket{d-1} \ot \ket{i-1},
    \label{isoDef2}
\end{align}
where $1 \leq i \leq d-1$ and $d \geq 3$. When $d=3$, the isometry above is
equivalent to $F_s$ in~\eqref{isoDef1} at $s=1/2$ and defines a channel pair
$(\NC_{1/2}, \NC_{1/2}^c)$.  In general $d \geq 3$ and the
isometry~\eqref{isoDef2} defines a pair of channels denoted by $\MC_d \colon
\hat \HC_a \mapsto \hat \HC_b$ and $\MC_d^c : \hat \HC_a \mapsto \hat \HC_c$. 
Let $U$ be any unitary on a $(d-1)$ dimensional Hilbert space.  Using $U$
and the standard basis on $\HC_a, \HC_b$, and $\HC_c$ we define unitary
operators $U_a = [0] \oplus U, U_b = U \oplus [d-1],$ and $U_c = U^*$ on
$\HC_a, \HC_b$, and $\HC_c$ respectively, here $*$ denotes complex conjugation
in the standard basis of $\HC_c$. The isometry $G$~\eqref{isoDef2} is
symmetrical in the sense that
\begin{align}
    G U_a = (U_b \ot U_c) G.
    \label{eq:symG}
\end{align}
Due to this symmetry, both channels, $\MC_d$ and $\MC_d^c$, have the property that
\begin{align}
    \MC_d (U_a \rho U_a^{\dag}) = U_b \MC_d (\rho) U_b^{\dag},  
    \quad
    \MC_d^c (U_a \rho U_a^{\dag}) = U_c \MC_d (\rho) U_c^{\dag},
    \label{eq:symMd}
\end{align}
where $\rho$ is any density operator on $\HC_a$.

The channels $\MC_d$ and $\MC_d^c$ have properties similar to those of $\NC_s$
and $\NC_s^c$, respectively. Like $\NC_s$, the channel $\MC_d$ has several degradable
sub-channels. Later in Sec.~\ref{sec:proofsChan} we show that a degradable
sub-channel of $\MC_d$ is obtained by restricting its input to operators on a
qubit subspace $\HC_{ai}$~\eqref{aiSub} or by restricting the input to a qubit
subspace obtained by applying $U_a$, defined above in \eqref{eq:symG}, to each
state in $\HC_{ai}$, where $i$ is some fixed number between $1$ and $d-1$.
Quantum states lying solely at the input of such a degradable qubit sub-channel
can be used to send an equal amount of quantum and private information to the
$\MC_d$ channel output $\HC_b$ but such states cannot be used to send any
quantum or private information to the $\MC_d^c$ output $\HC_c$.

Like $\NC_s^c$, the channel $\MC_d^c$ has a perfect sub-channel to its output $\HC_c$.  In
Sec.~\ref{sec:proofsChan}, we show this perfect sub-channel is obtained by
restricting the channel input to $\HC_{a'}$, a $(d-1)$-dimensional subspace of
$\HC_a$ spanned by $\{\ket{i}\}_{i=1}^{d-1}$. 
Quantum states lying solely in the channel input subspace $\HC_{a'}$ can be
used to perfectly send information to the $\MC_d^c$ output $\HC_c$, while
sending no information to the $\MC_d$ output $\HC_b$. 
This perfect transmission, along with arguments similar to those
above~\eqref{eq:NsCompCap}, can be used to show that the capacities
of $\MC_d^c$ satisfy
\begin{align}
    \QC(\MC_d^c)  = \PC(\MC_d^c) = \CC(\MC_d^c) = \log (d-1).
    \label{eq:MdCompCap}
\end{align}

A variety of channel classes were discussed in Sec.~\ref{sec:spChanCl}.  These
classes include degradable channels, anti-degradable channels, EB channels,
Hadamard channels, and less noisy channels. If a sub-channel of a channel does
not belong to any of these classes, then the channel itself does not belong to
that same class. Notice $\NC_{1/2}$ is a sub-channel of $\MC_d$ arising from a
restriction of the $\MC_d$ channel input to operators on a qutrit
subspace spanned by $\{\ket{0},\ket{1}, \ket{2}\}$. 
In Sec.~\ref{sec:NsChan}, we mentioned that $\NC_{1/2}$ is not degradable,
anti-degradable, EB, Hadamard, or less noisy.
As a result $\MC_d$ is not (anti)-degradable, EB, Hadamard, or less noisy. One
can easily verify that $\MC_d$, like $\NC_{1/2}$, is not a unital channel or a
PPT channel. This verification can be done by a direct computation. In the case
of ruling out PPT behaviour one can also use the fact that PPT channels have
zero quantum capacity, however $\NC_{1/2}$, a sub-channel of $\MC_d$ has
non-zero quantum capacity and thus $\MC_d$ also has non-zero quantum capacity.
An argument similar to the one above can be used to show that, similar to $\NC_s^c$, the channel $\MC_d^c$ is not (anti)-degradable, EB, Hadamard, less noisy, unital or PPT.

\subsection{The general $\OC$ channel}
\label{sec:glnChan}
\xb
\outl{Channel isometry $H$, $d_a = d_b = d$, $d_c = d-1$, unequal complex
Schmidt coefficients. Schmidt real and ordered. $F,G$ as special cases $H$,
Channel pair $\OC, \OC^c$. Repeat two key properties: degradable sub-channels
$\OC_{i}$ and perfect sub-channel to environment~(proofs at end of Sec.),
intuitive picture, outside the class of (anti)~degradable~(proof follows from
property), Unital~(simple cal), not EB/PPT~(proof $\NC_s$ later,
$\QC^{(1)}>0$), Hadamard~(EB argument)}
\xa

Both isometries in~\eqref{isoDef1} and~\eqref{isoDef2} can be viewed as special
cases of an isometry, $H \colon \HC_a \mapsto \HC_b \ot \HC_c$, where $d_a = d_b =
d$, $d_c = d-1$,
\begin{align}
    H \ket{0} = \sum_{j=0}^{d-2} \mu_j \ket{j} \ot \ket{j}, \quad 
    H \ket{i} = \ket{d-1} \ot \ket{i-1},
    \label{isoDef3}
\end{align}
$1 \leq i \leq d-1$, $\mu_i \in \Cbb$, and $\sum_i |\mu_i|^2 = 1$. Note that 
complex phases in $\mu_i$ can be absorbed in the definition of the standard
basis of $\HC_b$, so we may assume that $\mu_i$ are real and non-negative.  An
exchange of $\mu_i$ and $\mu_j$~($0 \leq i \neq j \leq d-2$) is equivalent to
an exchange of $\ket{i+1}$ and $\ket{j+1}$ in $\HC_a$ and an exchange of
$\ket{i}$ and $\ket{j}$ in both $\HC_b$ and $\HC_c$. All these exchanges can be
achieved by performing local unitaries on $\HC_a$, $\HC_b$, and $\HC_c$
respectively. As a result, we restrict our attention to $\{\mu_i\}$ arranged in
ascending order, i.e., $\mu_0 \leq \mu_1 \dots \leq \mu_{d-2}$.
In~\eqref{isoDef3}, the special case obained by setting $d = 3$, $\mu_0 =
\sqrt{s}$ and $\mu_1 = \sqrt{1-s}$ gives~\eqref{isoDef1}. On the other hand,
when $d$ is more general and all $\mu_i$ take equal values $1/\sqrt{d-1}$, we
obtain~\eqref{isoDef2}. 

The isometry $H$ generates a pair of channels $(\OC, \OC^c)$.  Like $\NC_s$ in
Sec.~\ref{sec:NsChan} and $\MC_d$ in Sec.~\ref{sec:MdChan}, the channel pair
$\OC$ has several degradable sub-channels.  Restricting the input of $\OC$ to
operators on $\HC_{ai}$~\eqref{aiSub}, $i$ is any fixed number between 1 to
d-1, results in a degradable sub-channel.  The channel complement $\OC^c$, like
$\NC_s^c$ in Sec.~\ref{sec:NsChan} and $\MC_d^c$ in Sec.~\ref{sec:MdChan}, has
a perfect sub-channel. This sub-channel has input dimension $d-1$ and is
obtained by restricting the $\OC^c$ channel input to operators on $\HC_{a'} =
\rm span \{\ket{i}\}_{i=1}^{d-1}$.  Proofs for both these properties are in
Sec.~\ref{sec:proofsChan}.

As mentioned earlier in Sec.~\ref{sec:NsChan} and Sec.~\ref{sec:MdChan},
properties of the type mentioned above imply that quantum states lying in the
qubit subspace $\HC_{ai}$ at the channel input can be used to send quantum and
private information to the $\OC$ channel output but not the $\OC^c$ output.
Additionally, quantum states lying solely in the $(d-1)$-dimensional $\HC_{a'}$
input subspace perfectly send all information to the $\OC^c$ output $\HC_c$.
Arguments similar to those above~\eqref{eq:NsCompCap} can be used to show that
\begin{align}
    \QC(\OC^c)  = \PC(\OC^c) = \CC(\OC^c) = \log (d-1).
    \label{eq:OcCap}
\end{align}

From~\eqref{eq:OcCap}, it follows that $\OC^c$ has non-zero quantum capacity
and thus $\OC^c$ does not belong to the class of less noisy channels. Using a
$\log$-singularity based argument~(see Sec.~\ref{sec:proofsChan}), one can show that
$\QC^{(1)}(\OC)>0$.  As a result, $\OC$ is also outside the class of less noisy
channels.
Using arguments similar to those used for $(\NC_s, \NC_s^c)$ in
Sec.~\ref{sec:NsChan} and for $(\MC_d, \MC_d^c)$ in Sec.~\ref{sec:MdChan}, one
can show that both $\OC$ and $\OC^c$ are not degradable, anti-degradable, EB,
Hadamard, unital, or PPT channels.

\subsection{Proofs of channel properties}
\label{sec:proofsChan}
\xb
\outl{Properties of $(\OC, \OC^c)$ imply properties of $(\NC_s, \NC_s^c)$
and $(\MC_d, \MC_d^c)$. Degradable sub-channels of $\OC$. Perfect sub-channel
of $\OC^c$}
\xa

\xb 
\outl{Degradable sub-channel claims about $\NC_s, \MC_d$ and $\OC$ follow from 
those of $\OC$. Proof of claim for $\OC$}
\xa

In Sec.~\ref{sec:NsChan}, below eq.~\eqref{isoDef1}, we claimed that $\NC_s$ has
degradable sub-channels obtained by restricting the channel input to $\HC_{ai}$
in~\eqref{aiSub}.
We made a similar claim about $\MC_d$, a $d$-dimensional generalization of
$\NC_{1/2}$ defined in Sec.~\ref{sec:MdChan}, and about $\OC$, a
$d$-dimensional generalization of both $\NC_s$ and $\MC_d$ defined in
Sec.~\ref{sec:glnChan}.
We prove all these claims here.  Instead of giving separate proofs for claims
about each channel $\NC_s,\MC_d$, and $\OC$, we give a single proof for the
$\OC$ channel. From this single proof, correctness of claims for $\NC_s$ and
$\MC_d$ follow as both $\NC_s$ and $\MC_d$ are special cases of $\OC$~(see
discussion in Sec.~\ref{sec:glnChan} below \eqref{isoDef3}).
Using this single proof, along with the definition of degradability and the
symmetry in eq.~\eqref{eq:symMd}, one can easily show an additional claim in
Sec.~\ref{sec:MdChan}: restricting the input of $\MC_d$ to a qubit sub-space
obtained by applying $U_{a}$, defined above eq.~\eqref{eq:symG}, to each state
in $\HC_{ai}$, results in a degradable sub-channel of $\MC_d$.

Let $\OC_i$ be a sub-channel of $\OC$ obtained by restricting the channel input
to a two-dimensional sub-space $\HC_{ai}$~(see eq.~\eqref{aiSub})
of $\HC_a$, where $i$ is a fixed integer between $1$ and $d-1$. This
sub-channel, $\OC_i$, is degradable. Prior to constructing a degrading map for
$\OC_i$, consider an isometry $H_i \colon \HC_{ai} \mapsto \HC_b \ot \HC_c$,
\begin{align}
    H_i \ket{0} = H \ket{0}, \quad
    H_i \ket{i} = H \ket{i},
\end{align}
where $H$ is defined in~\eqref{isoDef3}. This isometry $H_i$ generates the
$\OC_i$ sub-channel and its complement $\OC_i^c$. Let $\HC_e$ be a
$d$-dimensional Hilbert space.  For any fixed integer $i$ between $1$ and
$d-1$, consider another isometry $K_i\colon \HC_b \mapsto \HC_c \ot \HC_e$ of the
form
\begin{align}
    K_i \ket{j} = \ket{j} \ot \ket{j}, \quad \text{and} \quad
    K_i \ket{d-1} = \ket{i-1} \ot \ket{d-1},
\end{align}
where $0 \leq j \leq d-2$. The $\FC_i \colon \hat \HC_b \mapsto \hat \HC_c$
channel generated by the isometry $K_i$ satisfies
\begin{align}
    \FC_i \circ \OC_i = \OC_i^c.
    \label{degSub}
\end{align}
The equality above implies that $\FC_i$ is a degrading map for $\OC_i$ and thus
$\OC_i$ is degradable and $\OC_i^c$ anti-degradable.

\xb 
\outl{Perfect sub-channel claims about $\NC_s^c, \MC_d^c$ and $\OC^c$ follow
from those of $\OC^c$. Proof of claim for $\OC^c$}
\xa

In Sec.~\ref{sec:NsChan} we claimed that $\NC_s^c$ has perfect sub-channels
obtained by restricting the channel input to $\HC_{a'}$~(defined below
~\eqref{aiSub}).
A similar claim was made about $\MC_d^c$ in Sec.~\ref{sec:MdChan}, and about
$\OC^c$ in Sec.~\ref{sec:glnChan}.
All these claims are proven here by giving a single proof for the $\OC^c$
channel. From this proof, correctness of claims for $\NC_s^c$ and $\MC_d^c$
follow as both $\NC_s^c$ and $\MC_d^c$ are special cases of $\OC^c$~(see
discussion in Sec.~\ref{sec:glnChan} below \eqref{isoDef3}).
To prove that $\OC^c$ has a $(d-1)$-dimensional sub-channel which perfectly
maps its input to $\HC_c$, consider
\begin{align}
    \HC_{a'} = \Span \{\ket{1}, \ket{2}, \dots, \ket{d-1}\},
    \label{eq:HaPrimSub}
\end{align}
a $(d-1)$-dimensional sub-space of $\HC_a$ with projector $P_{a'} =
\sum_{i=1}^{d-1} [i]$.  Restricting the input of $\OC$ to $\HC_{a'}$ results in a
sub-channel $\OC'\colon \hat \HC_{a'} \mapsto \hat \HC_{c}$ with super-operator
\begin{align}
    \OC'(A) = VAV^{\dag},
    \label{subChanDef}
\end{align} 
where $V\colon \HC_{a'} \mapsto \HC_c$ is a bijection of the form
\begin{align}
    V\ket{i} = \ket{i-1}, \quad 1 \leq i \leq d-1,
    \label{vMap}
\end{align}
satisfying
\begin{align}
    VV^{\dag} = P_{a'}, \quad \text{and} \quad V^{\dag}V = I_c.
\end{align}
Since $V$ is a bijection, the sub-channel $\OC'$ in~\eqref{subChanDef}
perfectly maps its input space $\HC_{a'}$ to its $(d-1)$-dimensional output
$\HC_{c}$.

\section{Channel coherent information}
\label{sec:CI}
\xb
\outl{CI of $\NC_s, \MC_d$, and $\OC$. Proof}
\xa

The coherent information of $\NC_s, \MC_d$, and $\OC$ can be obtained from an
optimization of the form~\eqref{eq:q1Deg}. In general, this optimization is
non-trivial to carry out because the entropy bias~\eqref{entBias} is not
generally concave in $\rho$. As a result the coherent information $\QC^{(1)}$
for most channels remains unknown. However, one can show that optimizations for
$\QC^{(1)}$ of all three channels $\NC_s, \MC_d$, and $\OC$ can be reduced to a
one-parameter concave maximization over a bounded interval. For any $0 \leq s \leq 1/2$,
\begin{align}
    \QC^{(1)}(\NC_s) = \max_{\rho_a(u)} \Dl\big( \NC_s, \rho_a(u) \big),
    \label{eq:NsQ1}
\end{align}
where $\rho_a(u)$ is a one-parameter density operator of the form $\rho_a(u) =
(1-u)[0] + u [2]$ and $0 \leq u \leq 1$. For any $d \geq 3$,
\begin{align}
    \QC^{(1)}(\MC_d) = \max_{\rho_a(u)} \Dl\big( \MC_d, \rho_a(u) \big),
    \label{eq:MdQ1}
\end{align}
where $\rho_a(u) = (1-u)[0] + u [i]$, $0 \leq u \leq 1$ and $i$ is any fixed
integer between $1$ and $d-1$. For any $\mu_0 \leq \mu_1 \leq \dots \leq
\mu_{d-2}$,
\begin{align}
    \QC^{(1)}(\OC) = \max_{\rho_a(u)} \Dl\big( \OC, \rho_a(u) \big),
    \label{eq:OQ1}
\end{align}
where
\begin{align}
    \rho_{a}(u) = (1-u) [0] + u [d-1],
    \label{eq:optOp}
\end{align}
and $0 \leq u \leq 1$.  In~\eqref{eq:NsQ1},~\eqref{eq:MdQ1}, and~\eqref{eq:OQ1}
the maximization is over a density operator $\rho_{a}(u)$ which is supported
over a subspace of the $\HC_{ai}$ form~\eqref{aiSub}. In
Secs.~\ref{sec:NsChan},~\ref{sec:MdChan}, and~\ref{sec:glnChan}, we showed that
restricting the channel input to these $\HC_{ai}$ subspaces results in a
degradable channel. Since the entropy bias of a degradable channel is concave
in the channel input, each optimization in~\eqref{eq:NsQ1},~\eqref{eq:MdQ1},
and~\eqref{eq:OQ1} is a concave maximization. In addition, one can also show
that all three coherent informations ~\eqref{eq:NsQ1},~\eqref{eq:MdQ1},
and~\eqref{eq:OQ1} are strictly positive.

Since $\NC_s$ and $\MC_d$ are special cases of $\OC$~(see discussion
below~\eqref{isoDef2}), proving $\QC^{(1)}(\OC)>0$ and the
equality~\eqref{eq:OQ1} also proves $\QC^{(1)}(\NC_s)>0$, $\QC^{(1)}(\MC_d)>0$
and proves the equalities in both~\eqref{eq:NsQ1} and~\eqref{eq:MdQ1}. 
In what follows, Sec.~\ref{sec:Q1OChan} shows the equality~\eqref{eq:OQ1},
Sec.~\ref{sec:lSingCIPos} shows $\QC^{(1)}(\OC)>0$, Sec.~\ref{sec:subCh}
provides additional insight into certain sub-channels of $\OC$ and finally
Sec.~\ref{sec:Ren} extends our results to R\'{e}nyi entropies.

\subsection{Evaluating $\QC^{(1)}(\OC)$}
\label{sec:Q1OChan}

Let $\OC_{d-1}$ be the sub-channel of $\OC$ obtained by restricting the channel
input to the sub-space $\HC_{ai}$ in~\eqref{aiSub} where $i = d-1$.  
In this subsection, we will show that
\begin{align}
    \QC^{(1)}(\OC) \coloneqq \max_{\rho_a} \Dl(\OC, \rho_a)
    = \max_{0 \leq u \leq 1} \Dl(\OC_{d-1}, \rho_a(u))
    = \QC^{(1)}(\OC_{d-1}),
    \label{eq:q1ResB}
\end{align}
where $\rho_a(u)$ is defined in~\eqref{eq:optOp}.
The above equations can be proved in three steps. 
The first step exploits the structure of $(\OC, \OC^c)$ to show that 
restricting the input $\rho_a$ to a certain block-diagonal form preserves the 
optimal value in the second expression in~\eqref{eq:q1ResB}.  
The second step reduces this new optimization further to a one parameter 
problem~\eqref{eq:max1}, using a majorization-based approach. 
Finally, this one-parameter
problem is shown to be equivalent to finding the coherent information of
$\OC_{d-1}$.  These three steps are detailed as follows.  

Recall from ~\eqref{eq:HaPrimSub} that $\HC_{a'} = \Span \{\ket{1}, \ket{2}, \dots, \ket{d-1}\}$.  Let 
\begin{align}
    \HC_0 \coloneqq \Span \{\ket{0}\}, \quad \text{and} \quad \HC_1 \coloneqq \HC_{a'},
    \label{eq:APart}
\end{align}
so that $\HC_a = \HC_0 \oplus \HC_1$.
Let $\LC_{ij}$ be the space of linear operators from $\HC_{j}$
to $\HC_i$, i.e.,
\begin{align}
    \LC_{ij} = \{E_{ij} \; | \; E_{ij}\colon \HC_j \mapsto \HC_i\}.
    \label{LDef}
\end{align}
Any density operator on $\HC_a$ can be written as
\begin{align}
    \rho_a  = \bigoplus_{i,j} N_{ij},
    \label{eq:rhoForm}
\end{align}
where $i,j$ are binary and $N_{ij}$ are linear operators from $\HC_j$ to
$\HC_i$. 

\begin{remark}
\label{remIeJ}
Using \eqref{isoDef3}, it is straightforward to verify that,  
in the standard basis of $\HC_b$, 
$\OC(N_{ii})$ has zero off-diagonal entries, and 
if $i \neq j$, $\OC(N_{ij})$ has zero diagonal entries, while 
$\OC^c(N_{ij})$ is the null operator on $\HC_c$. 
\end{remark}

For step 1, starting  from any $\rho_a$, we obtain $\tilde{\rho}_a$ by
resetting $N_{ij} = 0$ for $i \neq j$.  We claim that $\Dl(\OC, \rho_a) \leq
\Dl(\OC, \tilde{\rho}_a)$.  To see this, let $\tilde{\rho}_b =
\OC(\tilde{\rho}_a)$ and $\tilde{\rho}_c = \OC^c(\tilde{\rho}_a)$.  Using the
above remark, $\tilde{\rho}_b$ can be obtained from $\rho_b$ by resetting all
the off-diagonal elements to 0, and $\tilde{\rho}_b$ is majorized by $\rho_b$
(see discussion below eq.~\eqref{eq:mj1} and Prb.II.5.5 in~\cite{Bhatia97}).
Applying Schur-concavity of the von-Neumann entropy~(see
Sec.~\ref{sec:majEntMin}), we have $S(\tilde{\rho}_b) \geq S(\rho_b)$.  
The above remark also implies that $\tilde{\rho}_c = \rho_c$, so 
$S(\tilde{\rho}_c) = S(\rho_c)$. 
Together, this proves the claim $\Dl(\OC, \rho_a) \leq \Dl(\OC, \tilde{\rho}_a)$.  
Thus, to maximize the entropy bias~\eqref{entBias} 
we can focus on $\rho_a$ of the form
in~\eqref{eq:rhoForm} where $N_{01} = N_{10} = 0$, i.e., 
\begin{align}
    \rho_a = (1-u) [0] + u \sg,
    \label{eq:rhoForm2}
\end{align}
where $0 \leq u \leq 1$ and $\sg$ is a density operator on $\HC_{a'}$. 

For step 2, note that the input $\rho_a$ in \eqref{eq:rhoForm2} gives the outputs 
\begin{align}
    \rho_b = (1-u) \left( \sum_{i=0}^{d-2} |\mu_i|^2 [i] \right) + u [d-1], 
\end{align}
and
\begin{align}
    \rho_c = (1-u) \left( \sum_{i=0}^{d-2} |\mu_i|^2 [i] \right) + u V \sg V^{\dag},
    \label{eq:rhoCForm}
\end{align}
where the channel parameters $\mu_0 \leq \mu_1 \leq \dots \leq \mu_{d-2}$ are
fixed.  Note that $S(\rho_b)$ only depends on $u$ while $S(\rho_c)$ depends on both
$u$ and $\sg$. Thus for any fixed $u$, the entropy bias $\Dl(\OC,
\rho_a)$ is maximum when $S(\rho_c)$ is minimum. 
We will prove in the next subsection that this minimum can always be attained
when $\sg = [d-1]$, for all relevant values of $u$ and the channel parameters
$\mu_0 \leq \mu_1 \leq \dots \leq \mu_{d-2}$.  As a result,
\begin{align}
    \QC^{(1)}(\OC) = \max_{0 \leq u \leq 1} \Dl\big( \OC, \rho_a (u) \big),
    \label{eq:max1}
\end{align}
where $\rho_a(u)$ is as given in~\eqref{eq:optOp}, obtained from setting 
$\sg = [d-1]$ in~\eqref{eq:rhoForm2}.  

To finish the proof, note that $\rho_a(u)$ is supported on $\HC_{a(d-1)}$, so,
expanding the maximization to a general density operator $\Lm$ supported on
$\HC_{a(d-1)}$ is nondecreasing: 
\begin{align}
    \max_{0 \leq u \leq 1} \Dl\big(\OC, \rho_a(u) \big) \leq
    \max_{\Lm} \Dl(\OC,\Lm) \,. 
    \label{eq:max2}
\end{align}
Since $\Lm$ is supported on $\HC_{a(d-1)}$, 
\begin{align}
    \max_{\Lm} \Dl(\OC,\Lm) = \max_{\Lm} \Dl(\OC_{d-1}, \Lm) \,,
    \label{eq:max3}
\end{align}
and the RHS of the above is $\QC^{(1)} (\OC_{d-1} )$.  
Combining \eqref{eq:max1}, \eqref{eq:max2}, and \eqref{eq:max3} thus gives
\begin{align}
    \QC^{(1)}(\OC) \leq \QC^{(1)}(\OC_{d-1}).
    \label{eq:max4}
\end{align}
The opposite inequality holds since $\OC_{d-1}$ is a sub-channel of $\OC$, 
which establishes~\eqref{eq:q1ResB}.

\subsection{Majorization and entropy minimization}
\label{sec:majEntMin}
\xb
\outl{Entropy min and majorization}
\xa

In this subsection, we prove the claim in the previous subsection that $\rho_c$
in \eqref{eq:rhoCForm} has minimum entropy when $\sg = [d-1]$, for any fixed
$u$ and channel parameters $\mu_0 \leq \mu_1 \leq \dots \mu_{d-2}$.  We first
summarize our main tools, majorization and Schur-concavity of the Shannon
entropy. Let us first consider majorization of real vectors.  Given a vector
$\xB$ in $\mathbb{R}^t$, let $\xB^{\downarrow}$ denote the vector obtained by
rearranging the entries of $\xB$ in descending order.  For any two vectors
$\xB,\yB$ in $\mathbb{R}^t$, we say that $\xB$ is majorized by $\yB$, in
symbols $\xB \prec \yB$, if the inequality,
\begin{align}
    \sum_{j=1}^{k} \xB_j^{\downarrow} \leq \sum_{j=1}^{k} \yB_j^{\downarrow},
    \label{eq:mj1}
\end{align}
holds for all $k \leq t$ and becomes an equality at $k=t$. The concept
of majorization can be generalized to Hermitian matrices as follows.
For any $t \times t$ Hermitian operator $N$, let $\vB(N)$ denote the vector of
singular values of $N$. For two Hermitian operators $N,M$ of the same size, 
we say that $N$ is majorized by $M$, $N \prec M$ in symbols, 
if $\vB(N) \prec \vB(M)$. The operator
$N+M$ satisfies~(see Ex.II.1.15 in~\cite{Bhatia97}),
\begin{align}
    \sum_{j=1}^{k} \vB_j^{\downarrow}(N+M) 
    \leq 
    \sum_{j=1}^{k} \vB_j^{\downarrow}(N)+
    \sum_{j=1}^{k} \vB_j^{\downarrow}(M)
\end{align}
for all integers $1 \leq k \leq d$, as a result,
\begin{align}
    \vB(N + M) \prec \vB^{\downarrow}(N) + \vB^{\downarrow}(M).
    \label{eq:majSum}
\end{align}
The Shannon entropy is Schur-concave, i.e., when a probability vector, $\pB$,
is majorized by another, $\qB$, then $h(\pB) \geq h(\qB)$. The von-Neumann
entropy of a density operator $\tau$ is equal to the Shannon entropy
of $\vB(\tau)$:
\begin{align}
    S(\tau) = h \big(\vB(\tau)\big).
    \label{eq:shVon}
\end{align}
Thus, like the Shannon entropy, the von-Neumann entropy is Schur-concave, i.e.,
if $\tau \prec \kappa$ then $S(\tau) \geq S(\kappa)$.

We now return to the problem of minimizing the 
von-Nuemann entropy $S(\rho_c)$.  Recalled from ~\eqref{eq:rhoCForm} 
that $\rho_c = (1-u) \left( \sum_{i=0}^{d-2} |\mu_i|^2 [i] \right) + u V \sg V^{\dag}$.  
Let $\sg' = V \sg V^{\dag}$ and 
$\Upsilon \coloneqq \sum_{i=0}^{d-2} |\mu_i|^2 [i]$, so, 
\begin{align}
    \rho_c = (1-u) \Upsilon + u \sg'.
    \label{eq:rhoCForm2}
\end{align}
Since $S(\rho_c)$ is a concave function, the minimum can be attained 
when $\sg'$ is a pure state.
Applying~\eqref{eq:majSum} to~\eqref{eq:rhoCForm2}, we obtain
\begin{align}
    \vB(\rho_c) \prec (1-u) \vB^{\downarrow}(\Upsilon) + u \vB^{\downarrow}(\sg')
    \label{eq:maj1}
\end{align}
where we use $\vB(cN) = c\vB(N)$ for any real $c>0$. Since $\sg'$ is a pure
state, the RHS of \eqref{eq:maj1} is a vector 
$\wB = ((1-u) |\mu_{d-2}|^2 + u, (1-u) |\mu_{d-3}|^2, \cdots, (1-u)|\mu_{0}|^2)$, 
and from Schur-concavity of the Shannon entropy, 
\begin{align}
    S(\rho_c) = h\big( \vB(\rho_c) \big) \geq h(\wB).
    \label{eq:maj2}
\end{align}
From the expression of $\wB$, equality is attained when $\sg' = [d-2]$.  
So the entropy of $\rho_c$ in~\eqref{eq:rhoCForm2} is
minimized when $\sg' = [d-2]$.  Using the expression of $V$ in \eqref{vMap}, 
$\sigma = [d-1]$. 

\subsection{Log-singularity and positivity}
\label{sec:lSingCIPos}
\xa
\outl{log-singularity, general argument for $\QC^{(1)}>0$, specific argument
for $\QC^{(1)}(\OC)>0$}
\xa

We shall be using an $\ep$ $\log$-singularity based method to show $\QC^{(1)}(\OC)>0$
for any $\mu_0 \leq \mu_1 \leq \dots \leq \mu_{d-2}$ and $d \geq 3$. For details
of the method see~\cite{Siddhu21}. 
Let $\rho(\ep)$ be a density operator that depends on a real parameter $\ep$.
The von-Neumann entropy, $S(\ep) = - \Tr \big( \rho( \ep) \log \rho (\ep)
\big) $ is said to have an $\ep$ $\log$-singularity if one or several
eigenvalues of $\rho(\ep)$ increase linearly from $0$ to leading order in
$\ep$. As a result of this singularity, $S(\ep) \simeq x |\ep \log \ep|$ where
$x>0$ is called the {\em rate} of the singularity.
The entropy bias, or coherent information of $\OC$ at $\rho_a(\ep)$, may be
concisely denoted by,
\begin{align}
    \Dl(\ep) \coloneqq S_b(\ep)- S_c(\ep),
    \label{eq:entBeps}
\end{align}
where $S_b(\ep)\coloneqq S\big( \rho_b(\ep) \big) $ and $S_c(\ep)= S\big( \rho_c(\ep)
\big)$. If an $\ep$ $\log$-singularity in $S_b(\ep)$ has larger rate than
the one in $S_c(\ep)$, and $\Dl(0) = 0$, then $\QC^{(1)}(\OC)>0$. In the present
case, let $\rho_a(\ep)$ be the density operator in~\eqref{eq:optOp} where $u$
is replaced with $\ep$. At $\ep = 0$, $\rho_a(\ep)$ is a pure state, hence
$\Dl(0) = 0$. In addition, $\rho_c(0)$ has rank $d_c = d-1$, and thus
$S_c(\ep)$ cannot have an $\ep$ $\log$-singularity. On the other hand $S_b(\ep)$
has an $\ep$ $\log$-singularity of rate $1$ and thus $\QC^{(1)}(\OC)>0$.

\subsection{Sub-channels}
\label{sec:subCh}
\xa
\outl{log-singularity, general argument for $\QC^{(1)}>0$, specific argument
for $\QC^{(1)}(\OC)>0$}
\xa

Let $\OC_i$ be the sub-channel of $\OC$ obtained by restricting the channel input
to operators on $\HC_{ai}$~\eqref{aiSub}, where $i$ is some fixed integer
between $1$ and $d-1$. 
Using arguments similar to those presented in subsection \ref{sec:Q1OChan}, 
one can show
\begin{align}
    \QC^{(1)}(\OC_i) = \max_{\rho_{ai}(u)}\Dl (\OC_i, \rho_{ai}\big( u \big) ),
    \label{eq:q1Oci}
\end{align}
where $\rho_{ai}(u)$ is a one-parameter density operator of the form
\begin{align}
    \rho_{ai}(u) = (1-u) [0] + u [i],
    \label{eq:rhoai}
\end{align}
and $0 \leq u \leq 1$. Let $S(\rho_{ci})$ be the entropy of $\rho_{ci} =
\OC^c_i(\rho_{ai})$.  Using majorization and Schur concavity arguments 
(similar to those leading to~\eqref{eq:maj2}) and
the fact that $\mu_0 \leq \mu_1 \leq \dots \leq \mu_{d-2}$ one can show that
\begin{align}
    S(\rho_{ci}) \geq S(\rho_{cj}),
    \label{eq:entC}
\end{align}
for all $i \leq j$. For any fixed $u$ in~\eqref{eq:rhoai} the entropy $S(\rho_{bi})$
of $\rho_{bi} = \OC_i(\rho_{ai})$, is independent of $i$.
Using~\eqref{eq:entC} and the definition of the entropy bias~\eqref{entBias}
one obtains,
\begin{align}
    \Dl \big( \OC_i, \rho_{ai}(u) \big) \leq
    \Dl \big( \OC_j, \rho_{aj}(u) \big) ,
    \label{eq:entbRel}
\end{align}
for any $i \leq j$. The relation above implies,
\begin{align}
    \QC^{(1)}(\OC_i) \leq \QC^{(1)}(\OC_j),
    \label{eq:q1ij}
\end{align}
for all $i \leq j$. Using the above equation along with~\eqref{eq:q1ResB} we
get
\begin{align}
    \QC^{(1)}(\OC_1) \leq \QC^{(1)}(\OC_2) \leq \dots \leq \QC^{(1)}(\OC_{d-1}) 
    = \QC^{(1)}(\OC)
    \label{eq:q1ij-2}
\end{align}

\subsection{R\'{e}nyi entropy}
\label{sec:Ren}
Before closing this section, we remark that to obtain~\eqref{eq:q1ResB} we used
(a) monotonicity of the von-Neumann entropy $S(\rho_b)$ under
block-diagonalization of $\rho_a$, and (b) concavity of the von-Neumann entropy
to argue that pure $\sg'$ minimizes $S(\rho_c)$, and finally, we utilized
majorization~(see Sec.~\ref{sec:majEntMin}) to argue that the maximize entropy
bias occurs when the input density operator has the form~\eqref{eq:optOp}. 
In the definition of~\eqref{entBias} if one replaces the von-Neumann entropy
with any R\'{e}nyi entropy, 
\begin{align}
    S_{\al}(\rho) = \frac{1}{1-\al} \log \Tr(\rho^\al),
    \label{eq:Renyi}
\end{align}
where $0 \leq \al \leq \infty$, then the corresponding
equation~\eqref{eq:q1ResB} would still hold. 
We outline the reasoning here.  First, monotonicity of the Schur-concave
R\'{e}nyi entropy is unaffected when $\rho_a$ is block-diagonalized.   
However, unlike the von-Neumann entropy, $S_{\al}$ is not concave for $\al >1$.
Thus to prove that the minimum R\'{e}nyi entropy
$S_{\al}(\rho_c)$~\eqref{eq:maj1} also occurs over pure states~$\sg'$ we employ
a different stratgey. Write 
\begin{align}
    S_{\al}(\rho_c) = S_{\al}(\NC(\sg'))
\end{align}
where the quantum channel $\NC$ acts as $\NC(\sg') = (1-u) \Upsilon
\Tr(\sg') + u \sg'$, with $0 \leq u \leq 1$. Next we use the fact that the
minimum output R\'{e}nyi entropy for $\NC(\sg')$ occurs at a pure state
input~(see Sec.II in~\cite{GourKemp17}).  With $\sg'$ being restricted to a
pure state, one can now use the majorization-based argument to show that the
R\'{e}nyi entropy $S_{\al}(\rho_c)$ is minimum when the input density operator
has the form~\eqref{eq:optOp}.  This majorization argument is unaffected when
the von-Neumann entropy is replaced by the R\'{e}nyi entropy which is 
also Schur-concave.

\section{Quantum capacity of $\OC$}
\label{sec:qcap}

Subject to the {\em spin-aligment conjecture}, introduced in Sec.~\ref{sec:spinAl},
we show that
\begin{align}
    \QC(\NC_s) = \QC^{(1)}(\NC_s), \quad
    \QC(\MC_d) = \QC^{(1)}(\MC_d), \quad \text{and} \quad
    \QC(\OC) = \QC^{(1)}(\OC).
    \label{eq:NsQ1Q}
\end{align}
To show these equalities above, we prove the third equality in~\eqref{eq:NsQ1Q}
and infer the other two because $\NC_s$ and $\MC_d$ are special cases of
$\OC$~(see Sec.~\ref{sec:glnChan}).

Our next step is to compute the coherent information of the channel $\OC^{\ot
n}\colon \hat \HC_a^{\ot n} \mapsto \hat \HC_b^{\ot n}$
\begin{align}
    \QC^{(1)}(\OC^{\ot n})= \max_{\rho} \Dl( \OC^{\ot n}, \rho ) 
    = \max_{\rho} [S\big( \OC^{\ot n} (\rho) \big) - S\big( (\OC^c)^{\ot n} (\rho) \big)],
    \label{enDiffMn}
\end{align}
where $n \geq 1$ and $\rho$ is a density operator on $\HC_a^{\ot n} = \HC_a^1
\ot \HC_a^2 \ot \cdots \ot \HC_a^n$, with the superscript $i$ indicating the
$i^\text{th}$ space. Using \eqref{eq:APart}, express
\begin{align}
    \HC_a^{\ot n} = (\HC_0 \oplus \HC_1)^{\ot n} = \bigoplus_{\bB \in \{0,1\}^n} \HC(\bB) 
\end{align}
where $\{0,1\}^n \coloneqq \{\bB \; | \; \bB(i) \in \{0,1\}\; , 1 \leq i \leq
n\}$ is the set of $n$ bit strings and $\HC(\bB) \coloneqq \HC_{\bB(1)} \ot
\HC_{\bB(2)} \ot \cdots \ot \HC_{\bB(n)}$. Let $\LC(\sB, \tB)$ be the space of
linear operators from $\HC(\tB)$ to $\HC(\sB)$. Any density operator $\rho$ on
$\HC_a^{\ot n}$ can be written as
\begin{align}
    \rho = \bigoplus_{\sB,\tB \in \{0,1\}^n} N(\sB, \tB),
    \label{rhoDefN}
\end{align}
for some $N(\sB, \tB) \in \LC(\sB, \tB)$. We will obtain the form of $\OC^{\ot
n} (\rho)$ and $(\OC^c)^{\ot n}(\rho)$ using the expression of $\rho$ above and
the following two remarks (which generalize Remark~\ref{remIeJ}).  
\begin{remark}
\label{remSeT}
If $\sB = \tB$ then $\OC^{\ot n} \big( N(\sB, \tB) \big)$ has zero off-diagonal
    entries in the basis $\{\ket{0}, \dots, \ket{d{-}1} \}^{\ot n}$ of
    $\HC_b^{\ot n}$.  If $\sB \neq \tB$ then $\OC^{\ot n} \big( N(\sB, \tB)
    \big)$ has zero diagonal entries in the aforementioned basis of $\HC_b^{\ot
    n}$, while $(\OC^c)^{\ot n} \big( N(\sB, \tB) \big) = 0$, the null
    operator.
\end{remark}

To prove remark \ref{remSeT}, we first express any operator $N(\sB, \tB)$ as a
linear combination of operators of the form,
\begin{align}
    F = \bigotimes_{k = 1}^{n} F_{\sB(k) \tB(k)}, \quad \text{where} \quad
    F_{\sB(k) \tB(k)} \in \LC_{\sB(k) \tB(k)}.
    \label{simpleOn}
\end{align}
To illustrate these notations, for example, take $n=3, \sB = (0,1,1), \tB =
(1,0,1)$.  Then, each of these operators $F$ has the form $F = F_{01} \otimes
F_{10} \otimes F_{11}$ where $F_{01}:\HC_1 \mapsto \HC_0$, $F_{10}:\HC_0
\mapsto \HC_1$, and $F_{11}:\HC_1 \mapsto \HC_1$.  

Note that the properties in remark \ref{remSeT} are preserved under linear
combinations, so, it suffices to prove the remark for an operator $F$ in
\eqref{simpleOn} instead of $N(\sB, \tB)$.  The channels $\OC^{\ot n}$ and
$(\OC^c)^{\ot n}$ act on $F$ in \eqref{simpleOn} as follows,
\begin{align}
    \OC^{\ot n} (F) & = \bigotimes_{k=1}^n \OC \big( F_{\sB(k) \tB(k)} \big), \quad \text{and} \label{Bn} \\
    (\OC^c)^{\ot n} (F) & = \bigotimes_{k=1}^n \OC^c \big( F_{\sB(k) \tB(k)} \big) \label{Cn} \,. 
\end{align}

\begin{proof} (of remark \ref{remSeT})
    If $\sB = \tB$, for each $k$, let $\sB(k) = \tB(k) = i_k \in \{0,1\}$.  So,
    $F_{\sB(k) \tB(k)} = F_{i_k i_k}$.  From Remark~\ref{remIeJ}, $\OC(F_{i_k
    i_k})$ is diagonal in the standard basis of $\HC_b$, so, from~\eqref{Bn},
    $\OC^{\ot n}(F)$ is also diagonal in the standard basis of $\HC_b^{\ot n}$.
    If $\sB \neq \tB$, they differ at some position $k$, so $F_{\sB(k) \tB(k)}
    = F_{ij}$ for $i \neq j$.  From Remark~\ref{remIeJ}, $\OC(F_{\sB(k)
    \tB(k)}) = \OC(F_{ij})$ has zero diagonal entries, so from~\eqref{Bn},
    $\OC^{\ot n}(F)$ also has zero diagonal entries.
    Furthermore, from Remark~\ref{remIeJ} $\OC^c(F_{ij}) = 0$, so
    from~\eqref{Cn}, $(\OC^c)^{\ot n}(F) = 0$ also.
\end{proof}

Using \eqref{rhoDefN}, we write
\begin{align}
    \OC^{\ot n} (\rho) &= \sum_{\sB \neq \tB} \OC^{\ot n}\big( N(\sB,\tB) \big)
    + \sum_{\sB} \OC^{\ot n}\big( N(\sB,\sB) \big) \label{BnOut} \\
    (\OC^c)^{\ot n} (\rho) &= \sum_{\sB \neq \tB} (\OC^c)^{\ot n}\big(
    N(\sB,\tB) \big) + \sum_{\sB} (\OC^c)^{\ot n}\big( N(\sB,\sB) \big)
    \label{CnOut}
\end{align}
From remark~\ref{remSeT}, the first term on the right side 
of the equality in~\eqref{BnOut} has zero diagonal entries, the second
term has zero off-diagonal entries, while the first term on the right side of 
the equality in \eqref{CnOut} is zero. 
Therefore, setting $N(\sB,\tB) = 0 $ for all $\sB \neq
\tB$ has no effect on $(\OC^c)^{\ot n}(\rho)$, nor on the diagonal entries of 
$\OC^{\ot n}(\rho)$, while all off-diagonal
entries of $\OC^{\ot n}(\rho)$ become zero's.  This may increase 
$S\left(\OC^{\ot n}(\rho)\right)$ but not $S\left((\OC^c)^{\ot n}(\rho)\right)$
so, when maximizing the entropy difference in the second equality in~\eqref{enDiffMn},
we may restrict to density operators of the form
\begin{align}
    \rho = \bigoplus_{\sB \in \{0,1\}^n} N(\sB, \sB).
    \label{rhoRes}
\end{align}

To proceed with the analysis, we re-express \eqref{rhoRes} as a specific convex
combination of states.  Let $M$ denote a subset of $\{1,\dots,n\}$ and $M^c$
the 
complement of $M$ in $\{1,\dots,n\}$. Let $|M|$ and $|M^c|$ denote the sizes of
$M$ and $M^c$.  
We may use the subset $M$ to label some of the channel uses, or some 
of the input systems, or some of the output systems.  
For any such subset $M$, let $\om_M$ denote a density operator acting on   
the corresponding subset of input spaces $\otimes_{i \in M} \HC_a^i$, 
where $\HC_a^i$ denotes the input space of the $i$-th channel use,  
and let $\dya{0}^{\ot M^c}$ denote the pure state $(\dya{0})^{\otimes |M^c|}$
on the complement set of input spaces $\otimes_{j \in M^c} \HC_a^j$.    
Using this notation, we now show that 
the density operator in~\eqref{rhoRes} can be written as
\begin{align}
    \rho =  \sum_{M} x_M \; \om_{M} \ot \dya{0}^{\ot M^c}.
    \label{rhoResTen}
\end{align}
for some density operators $\om_{M}$ and for some 
$x_M$ between zero and one such that
\begin{align}
    \sum_M x_M = 1 \,.
    \label{eq:x-distr-alt}
\end{align}
In the above and throughout, the summation over $M$ is over all subsets of $\{1,\dots,n\}$. 
For an arbitrary $\sB$, let the $k$-th entry of $\sB$ be $i_k$.  Recall that
$N(\sB, \sB)$ is a linear combination of operators of the form $F_{i_1 i_1}
\otimes F_{i_2 i_2} \otimes \dots \otimes F_{i_n i_n}$.  Recall also that
$\HC_0 = \Span \{\ket{0}\}$ so $F_{00} \propto \dya{0}$ and for each $k$ where
$i_k=0$, $F_{i_k i_k} \propto \dya{0}$.  It means that $N(\sB, \sB) = \eta_M
\otimes \dya{0}^{M^c}$ where $M = \{i \in \{1,\dots,n\}: s_i = 1\}$ and $M^c =
\{i \in \{1,\dots, n\}: s_i = 0\}$ and $\eta_M$ is an operator acting on
$\otimes_{i \in M} \HC_a^i$.  Substituting this into the expression for $\rho$
in \eqref{rhoRes}, and applying $\rho \geq 0$ and $\tr\rho = 1$ gives $\eta_M
\geq 0$ for each $M$, $\sum_M \tr \eta_M = 1$ so we can write $\eta_M = x_M
\om_M$ for $x_M$ between 0 and 1, $\sum_M x_M = 1$, and each $\om_M$ is a
density operator.

Using the form of $\rho$ in \eqref{rhoResTen}, one may write
\begin{align}
    \OC^{\ot n}(\rho) = 
    \sum_M x_M \; (\, \dya{d{-}1}^{\ot |M|} \, )_M \ot \Upsilon^{\ot M^c},
    \label{BOut1}
\end{align}
and
\begin{align}
    (\OC^c)^{\ot n}(\rho) = 
    \sum_{M} x_M \; \om_M' \ot \Upsilon^{\ot M^c},
    \label{COut1}
\end{align}
where $\om_M' = V \om_M V^{\dag}$ and $\Upsilon = \sum_{i=0}^{d-2} |\mu_i|^2
[i]$ (see \eqref{isoDef3} and \eqref{vMap}).  For fixed channel parameters
$\{\mu_i\}$, the entropy $S\big(\OC^{\ot n}(\rho) \big)$ only depends on
$\{x_{M}\}$ while the entropy $S \big( (\OC^c)^{\ot n}(\rho) \big)$ depends on
both $\om_{M}$ and $\{x_{M} \}$. Thus for any fixed $\{\mu_i\}$ and $\{x_{M}
\}$, the maximum entropy difference~\eqref{enDiffMn} occurs when $S \big(
(\OC^c)^{\ot n}(\rho) \big)$ is minimum. This minimum can be obtained by
solving the optimization problem
\begin{align}
    \label{qnMinP1}
    \min & \; S(\kappa)  \\
    \kappa &= \sum_{M} x_M \; \om_M' \ot \Upsilon^{\ot M^c},
    \nonumber \\
    \om_{M}' &\geq 0, \nonumber \\
    \Tr(\om_{M}') &= 1, \forall M \subset \{1,\dots,n\} \nonumber,
\end{align}
where $\Upsilon$ only depends on the fixed $\{\mu_i\}$ and $\{x_M\}$ is also
held fixed.  In Sec.~\ref{sec:spinAl} we introduce the {\em spin alignment
conjecture}. Using this conjecture, the minimization problem above has a close
form solution $\om_{M} = (\dya{d-1})^{\ot |M|}$ for any fixed $\{\mu_i\}$ and
$\{x_M\}$. This closed form solution implies that the maximum entropy
difference~\eqref{enDiffMn} is attained on an input density operator of the
form
\begin{align}
    \rho =  \sum_{M} x_M \; \om_M \ot (\dya{0})^{\ot M^c},
    \label{rhoNForm}
\end{align}
where $\om_{M} = (\dya{d-1})^{\ot |M|}$.  The above density operator is
supported on a subspace $\HC_{ai}^{\ot n}$, $i = d-1$ of $\HC_a^{\ot n}$. Using
arguments similar to those in below eq.~\eqref{eq:max1} the maximum entropy
difference~\eqref{enDiffMn},
\begin{align}
    \QC^{(1)}(\OC^{\ot n})= \max_{\rho} \Dl( \tilde \OC^{\ot n}, \rho ) =
    \QC^{(1)}(\OC_{d-1}^{\ot n}).
    \label{enDiffMn2}
\end{align}
In Sec.~\ref{sec:proofsChan}, we showed that this
sub-channel $\OC_{d-1}$, is degradable. Since $\OC_{d-1}$ is degradable,
$\QC^{(1)}(\OC_{d-1}^{\ot n}) = n \QC^{(1)}(\OC_{d-1})$, and thus the 
equality~\eqref{enDiffMn2} simplifies
\begin{align}
    \QC^{(1)}(\OC^{\ot n}) = n\QC^{(1)}(\OC_{d-1}).
    \label{Add}
\end{align}
The equality along with \eqref{eq:q1ResB} gives the desired result
\begin{align}
    \QC^{(1)}(\OC^{\ot n}) = n\QC^{(1)}(\OC).
    \label{Additivity}
\end{align}

\section{Spin alignment conjecture}
\label{sec:spinAl}
Since this conjecture may be of independent interest, it is presented in a
more self-contained manner, with occasional repetition of information from the
previous section.  
Let $\ket{0}$ and $\ket{1}$ denote the spin up and spin down states of a 
spin-$\frac{1}{2}$ particle (we just call this a spin). Let
\begin{align}
    Q = s \, \dya{0} + (1-s) \, \dya{1},
    \label{qVar}
\end{align}
be a fixed mixed state of a spin, where $0 \leq s \leq 1/2$. 
For each $M \subset \{1,2,\cdots,n\}$, let $M^c$ be the complement of $M$, and let
$|M|$ and $|M^c|$ be the sizes of $M$ and $M^c$, respectively.
Consider $n$ spins, and view $M$ as a subset of the spins.  
We use $\om_{M} \ot Q^{\ot M^c}$ to denote a state on these $n$ spins where
each spin in $M^c$ is in the state $Q$, and the spins in $M$ are in a joint
state given by the density matrix $\om_{M}$.    
Let $\{x_M\}$ be non-negative numbers such that  
\begin{align}
    \sum_{M} x_M = 1 \,, 
    \label{eq:x-distr}
\end{align}
where the sum is over all possible subsets $M$ of $\{1,\dots,n\}$.  For an
arbitrary set of such numbers $\{x_M\}$, consider $n$ spins in the state
$\kappa =  \sum_{M} x_M \; \om_{M} \ot Q^{\ot M^c}$ where $\om_{M}$ are
variables.  The goal is to minimze the von Neuman entropy of $\kappa$,
$S(\kappa) = -\Tr(\kappa \log \kappa)$.   Formally, the {\bf entropy
minimization problem} is given by 
\begin{align}
    \label{qMinGln}
    \min & \; S(\kappa) \\
    \kappa &=  \sum_{M} x_M \; \om_{M} \ot Q^{\ot M^c}, \nonumber \\
    \om_{M} & \geq 0, \nonumber \\
    \Tr(\om_{M}) &= 1. \nonumber
\end{align}
We conjecture that the entropy minimization problem 
has an optimal solution when all spins align with one another as 
much as possible and aligned with the eigenstate corresponding to the 
maximum eigenvalue of $Q$; in other words, for all $2^n - 1$
possible non-empty subsets $M$,  
\begin{align}
    \label{qMinGlnSol}
    \om_M = \dya{1}^{\otimes |M|} \,.
\end{align}
Note that the minimum of \eqref{qMinGln} can be attained 
on pure states $\omega_M$ due to concavity of the von Neumann
entropy. 

We have not come upon a general proof of this conjecture but we have found proofs
for various special cases and numerical evidence in other cases. In what
follows we briefly mention this evidence for our conjecture.  

\subsection{Special case $n = 1$}
\label{sec:spinAln1}
For 1 spin, the entropy minimization problem~\eqref{qMinGln} takes the form
\begin{align}
    \label{q1MinP1}
    \min & \; S(\kappa) \\
    \kappa &=  x_{1} \om_{1} + x_{\phi} Q, \nonumber \\
    \om_{1} &\geq 0, \nonumber \\
    \Tr(\om_{1}) &= 1, \nonumber
\end{align}
where $\phi$ denotes the null set, the subscript $1$ is a shorthand for
$M=\{1\}$, and $x_1$ and $x_{\phi}$ are arbitrary but fixed non-negative
numbers that sum to $1$.  The above optimization problem (\ref{q1MinP1}),
generalized to any dimension and any valid density matrix $Q$, can be solved
using Schur-concavity of the von Neumann entropy along with a majorization
inequality, see eq.~(II.16) on pg.~35 in~\cite{Bhatia97}.  The optimal solution
$\om_{1}$ can always be chosen to be the projector onto any 1-dimensional
eigenspace of $Q$ corresponding to its maximum eigenvalue.  This proves our
conjecture for $n=1$.  

\subsection{Special case $n=2$ and $s = 1/2$}
When there are two spins and $s=\frac{1}{2}$, $Q = \frac{I}{2}$, the
optimization problem in~\eqref{qMinGln} takes the form
\begin{align}
    \label{q2MinP2}
    \min & \; S(\kappa)  \\
    \kappa &=  x_{12} \; \om_{12} + x_{1} \; \om_1 \ot \frac{I}{2} + x_2 \; \frac{I}{2} \ot \om_2 +
    x_{\phi} \; \frac{I}{2} \ot \frac{I}{2} \nonumber \\
    \om_{1} &\geq 0, \Tr(\om_{1}) = 1, \nonumber \\
    \om_{2} &\geq 0, \Tr(\om_{2}) = 1, \nonumber \\
    \om_{12} &\geq 0, \Tr(\om_{12}) = 1, \nonumber
\end{align}
where $x_1, x_2,x_{12}$ and $x_{\phi}$ are arbitrary but fixed non-negative
numbers that sum to $1$.  There is enough symmetry to assume without loss of
generality that $\om_1 = \dya{1}, \om_2 = \dya{1}$ and using the optimal
solution for the high-dimensional generalization of the $n=1$ case
(\ref{q1MinP1}), $\om_{12} = \dya{1} \ot \dya{1}$ can be shown to be an optimal
solution. 

We mention in passing that when $n=2$ but $s$ is arbitrary, if we assume
$\om_{1}, \om_{2}, \om_{12}$ to be diagonal, we can also show the optimality of
the conjectured solution, by a detailed case analysis involving majorization
and Schur-concavity of the von Neumann entropy.

\subsection{Numerical evidence}

Our conjecture is backed by numerical evidence that we gathered for the
optimization problem \eqref{qMinGln} with $n=2,\dots ,6$.  We randomly sampled
probability distributions $\lbrace x_M\rbrace_{M}$ (see
\eqref{eq:x-distr}), restricted the optimization to pure states $\omega_M$ (see
comment after \eqref{qMinGlnSol}) in order to reduce the number of free
parameters, and used both gradient descent-based and global optimization
techniques (particle swarm optimization).  In all instances, the optimization
converged to the conjectured solution in \eqref{qMinGlnSol}, i.e., with all
states $\omega_M$ being tensor products of (a rotated version of) the
eigenvector of $Q$ corresponding to its largest eigenvalue.  
Of course, it is possible that such convergence is to a local minimum rather
than a global minimum, however we have not found minima with values smaller
than the conjectured value.

\subsection{Generalization to higher dimensions}\label{sec:conjecture-d-dim}

The entropy minimization problem~\eqref{qMinGln} can be generalized to higher
spins (a spin-$\frac{1}{2}$ particle is a qubit and higher-spin particles
correspond to qudits): 
\begin{align}
    \label{qMinGlnd}
    \min & \; S(\kappa) \\
    \kappa &=  \sum_{M} x_M \; \om_{M} \ot \Upsilon^{\ot M^c}, \nonumber \\
    \om_{M} &\geq 0, \nonumber \\
    \Tr(\om_{M}) &= 1, \nonumber
\end{align}
where $\Upsilon$ is a fixed density matrix on a qudit, and $\{x_M\}$ a fixed
distribution on the subsets of $\{1,\dots,n\}$.  As before, the minimum can be
attained on pure states $\omega_M$ due to concavity of the von Neumann entropy. 
Let $|\gamma\rangle$ be an eigenvector of $\Upsilon$ corresponding to a maximum
eigenvalue.  We conjecture that the entropy is minimized when all the spins are
aligned with one another as much as possible; that is, 
\begin{align}
    \label{qMinGlnSol-gen}
    \om_M = \dya{\gamma}^{\otimes |M|},
\end{align}
for all $2^n-1$ possible non-empty subsets $M$.  

\subsection{R\'{e}nyi-2 entropy variation of the conjecture holds}
\label{sec:renyi-2-entropy}

We discuss a side result that illustrates the intuition behind the spin
alignment conjecture.  Unless otherwise stated, symbols defined here are used
only for this subsection.  
Consider a variation of the minimization problem \eqref{qMinGlnd} wherein the
von Neumann entropy is replaced by the R\'{e}nyi-2 entropy, which according
to~\eqref{eq:Renyi}, is given by 
\begin{align}
    S_2(\rho) = - \log \Tr(\rho^2) \,.
    \label{eq:Renyi2}
\end{align}
In this case, the minimum is attained by states given by
\eqref{qMinGlnSol-gen}, and we will outline the proof in this subsection.

Using the expression for R\'{e}nyi-2 entropy above, the goal to minimize
$S_2(\kappa)$ in \eqref{qMinGlnd} is equivalent to maximizing
\begin{align}
    \Tr(\kappa^2)
    = \sum_{M M'} x_M x_{M'} \; \Tr \left[
    \left( \om_{M} \ot \Upsilon^{\ot M^c} \right) \left( \om_{M'} \ot \Upsilon^{\ot M'^{c}} \right) \right]
    \label{eq:kappa-purity}
\end{align}
where we have used the expression for $\kappa$ in \eqref{qMinGlnd}.  Since $x_M
x_{M'} \geq 0$ for all $M,M'$, it suffices to show that the states given by
\eqref{qMinGlnSol-gen} maximize \emph{each} 
\begin{equation}
  \Tr \left[ \left( \om_{M} \ot \Upsilon^{\ot M^c} \right) \left( \om_{M'} \ot \Upsilon^{\ot M'^{c}} \right) \right]
\label{eq:newmax}
\end{equation}  
First, note that \eqref{eq:newmax} can be maximized on pure $\om_{M}, \om_{M'}$
due to convexity.  
Second, we can group the $n$ spins into $4$ systems, where systems $1,2,3,4$ 
include spins in $M \cap M'^{c}$, $M \cap M'$, $M^c \cap M'$, and 
$M^c \cap M'^{c}$ respectively.  
Third, the trace over system $4$ contributes a multiplicative constant to
\eqref{eq:newmax} so we can focus on the maximization on systems $1,2,3$.  
The conjecture can be proved if the following lemma holds.
\begin{lemma}
    Consider a tripartite system with parts $1,2,$ and $3$, and define a
    maximization problem, 
    \begin{align}
        \max_{|\mu\rangle, |\nu\rangle} \Tr  \left[ \left( |\mu\rangle
    \langle\mu|_{12} \otimes \Th_3 \right)  \left( Z_1 \otimes |\nu \rangle
    \langle \nu|_{23} \right) \right],
    \end{align}
    where the subscripts $1,2,3$ indicate which systems are acted on
    by the operators, $\Th,Z$ are positive semidefinite operators, and
    $|\mu\rangle, |\nu\rangle$ are unit vectors. The problem above attains its
    maximum at 
    \begin{align}|\mu\rangle_{12} = |\th \rangle_1 \otimes |\xi \rangle_2, \quad
    |\nu\rangle_{23} = |\xi \rangle_2 \otimes |\zt \rangle_3,
    \end{align}
    where $|\th \rangle$ is any eigenvector of $\Th$ corresponding to its
    maximum eigenvalue, $|\zt \rangle$ is any eigenvector of $Z$
    corresponding to its maximum eigenvalue, and $|\xi \rangle$ is any vector
    in system $2$. 
\end{lemma}

\begin{proof}
We can take the computational bases of systems $1$ and $3$ to be the eigenbases
    of $\Theta$ and $Z$ respectively, so they are diagonal without loss of
    generality.
Let $\Theta = \sum_j \theta_j |j\rangle \langle j|$, $Z = \sum_i \zeta_i
    |i\rangle \langle i|$ be their eigen-decompositions.  We can always express
    $|\mu\rangle = \sum_i a_i |i\rangle_1 \otimes |\alpha_i\rangle_2$ and
    $|\nu\rangle = \sum_j b_j  |\beta_j \rangle_2 \otimes |j\rangle_3$ for some
    nonnegative amplitudes $a_i$, $b_j$ and unit vectors $|\alpha_i\rangle$ and
    $|\beta_j \rangle$ on system $2$, such that $\sum_i a_i^2 = \sum_j b_j^2 =
    1$.  Then, 
\begin{align}
& \Tr  \left[
\left( |\mu\rangle \langle\mu|_{12} \otimes \Theta_3 \right)  
\left( Z_1 \otimes |\nu \rangle \langle \nu|_{23} \right) \right]
\nonumber 
\\
& =  \Tr  \left[
\left( \sum_{i_1,i_2} a_{i_1} a_{i_2} |i_1\rangle \langle i_2|_1 
       \otimes |\alpha_{i_1}\rangle  \langle \alpha_{i_2}|_2 
       \otimes \sum_j \theta_j |j\rangle \langle j|_3 \right)  
\left( \sum_i \zeta_i |i\rangle \langle i|_1 \otimes 
       \sum_{j_1,j_2}  b_{j_1} b_{j_2} |\beta_{j_1} \rangle \langle \beta_{j_2}|_2 \otimes 
       |j_1 \rangle \langle j_2|_3
\right) \right]
\nonumber 
\\
& =
 \sum_i \sum_j a_i^2 \zeta_i \; b_j^2 \theta_j  
       | \langle \alpha_{i} | \beta_{j} \rangle|^2  
\nonumber 
\\
& \leq
 \max_{i,j} \zeta_i \theta_j  
       | \langle \alpha_{i} | \beta_{j} \rangle|^2  \leq ~ \max_{i,j} \zeta_i \theta_j \,.
\end{align}
In the last line above, we use a convexity argument noting that $\{a_i^2
    b_j^2\}_{i,j}$ is a probability distribution.  The proposed solution
    $|\mu\rangle_{12} = |\delta\rangle_1 \otimes |\xi \rangle_2$,
    $|\nu\rangle_{23} = |\xi \rangle_2 \otimes |\gamma \rangle_3$ attains the
    upper bound $\max_{i,j} \zeta_i \theta_j$ thus must be an optimal solution
    for the maximization problem, proving the lemma. 
\end{proof}

\section{Quantum capacity bounds}\label{sec:qcap-bounds}

\subsection{Upper bound on $\QC(\NC_s)$}

We now derive an analytic upper bound and a tighter numerical upper bound
on the quantum capacity of $\NC_s$. The analytic bound matches the SDP bound in
Prop.~16 of~\cite{WangXieEA17} where it appears without proof. We provide a
proof showing the SDP bound matches the well-known ``transposition
bound'', which states that $\QC(\BC) \leq \log \|\TC\circ\BC\|_\diamond$ for any channel $\BC$~\cite{HolevoWerner01}.
Here, $\TC\colon X\mapsto X^T$ denotes the
transposition map, taken with respect to the same basis used to define the
maximally entangled state $|\phi\rangle$ in \eqref{eq:maxEnt}.  For a
superoperator $\Psi\colon \hat{\HC}\to \hat{\HC'}$ the diamond norm
$\|\Psi\|_\diamond$ is defined as 
\begin{align}
    \|\Psi\|_\diamond = \sup \lbrace\| (\IC_{\HC}\otimes \Psi)(X)\|_1\colon X\in\hat{\HC}\otimes\hat{\HC}, \|X\|_1\leq 1 \rbrace.
\end{align}
The diamond norm of any linear superoperator can be computed by a semidefinite
program \cite{Watrous13}, which, for $\TC\circ\BC$, is given by
\begin{align}
    \begin{aligned}
    \|\TC\circ\BC\|_\diamond = \text{min.~} &\frac{1}{2}(\|Y_{a}\|_\infty + \|Z_{a}\|_\infty)\\
    \text{s.t.~} & Y_{ab}, Z_{ab}\geq 0\\
    &\begin{pmatrix}
    Y_{ab} & -\TC_b(J_{ab}^\BC)\\
    -\TC_b(J_{ab}^\BC) & Z_{ab}
    \end{pmatrix}\geq 0 
    \end{aligned}
\label{eq:diamond-norm-sdp}
\end{align}
where $J^\BC_{ab}$ denotes the unnormalized Choi-Jamio\l{}kowski operator of
$\BC$~\eqref{eq:cjOp}, and $\TC_b = \IC_{a}\otimes \TC$ denotes the partial
transpose with respect to system $b$, i.e., the transpose map $\TC$ acts on
$\hat{\HC}_b$.

We compare the transposition bound to another bound on $\QC(\BC)$ by Wang et
al.~\cite{WangFangEA18} defined in terms of a quantity $\Gamma(\BC)$, which is
the solution of the following semidefinite program:
\begin{align}
    \begin{aligned} 
    \Gamma(\BC) = \text{max.~} & \tr R_{ab}J^\BC_{ab}\\
    \text{s.t.~} & R_{ab},\rho_a \geq 0\\
    & \tr\rho_a = 1\\
    & -\rho_a\otimes \one_b \leq \TC_b(R_{ab}) \leq \rho_a\otimes \one_b
    \end{aligned}
\label{eq:gamma-sdp}
\end{align}
The two bounds on the quantum capacity of $\BC$ are related as follows:
\begin{proposition}[Holevo, Werner~\cite{HolevoWerner01}, Wang et
    al.~\cite{WangFangEA18}]\label{thm:qc-upper-bounds} For any quantum channel
    $\BC$,
	\begin{align}
	\QC(\BC) \leq \log \Gamma(\BC) \leq \log \|\TC\circ\BC\|_\diamond,
	\end{align}
	where $\Gamma(\BC)$ is defined in \eqref{eq:gamma-sdp}.
\end{proposition}

The two upper bounds mentioned above yield an analytical upper bound on the
quantum capacity of the channel $\NC_s$. 
Due to unitary
equivalence of the $\NC_s$ channel to the one in~\cite{WangXieEA17}~(see
discussion below eq.~\eqref{isoDef1}), the SDP bound $\log \Gamma(\NC_s)$
matches the one stated without proof in Prop.~16 of~\cite{WangXieEA17}.

\begin{theorem}
	\label{thm:Q}
	For $s\in[0,1/2]$ we have $\log\Gamma(\NC_s) = \log \|\TC\circ\NC_s\|_\diamond = \log(1+\sqrt{1-s})$, and hence
	\begin{align}
	\QC(\NC_s) \leq \log(1+\sqrt{1-s}).
	\end{align}
\end{theorem}
\begin{proof}
    The theorem is proved by asserting that \begin{align} 1+\sqrt{1-s}\leq
    \Gamma(\NC_s)\leq \|\TC\circ\NC_s\|_\diamond \leq
    1+\sqrt{1-s},\label{eq:Qcap-bounds} \end{align} from which the claim
    follows via Proposition~\ref{thm:qc-upper-bounds}.
	
    To prove the first inequality in \eqref{eq:Qcap-bounds}, we pick the
    following operators $(R_{ab},\rho_a)$ in the SDP \eqref{eq:gamma-sdp} for
    $\Gamma(\NC_s)$:
    \begin{align}
        \rho_a &= \frac{1}{2} ( [0]_a + [2]_a)\\
        R_{ab} &= \frac{1}{2}([00]_{ab} + [01]_{ab} + |01\rangle\langle 22|_{ab} + |22\rangle\langle 01|_{ab} + [22]_{ab}).
	\end{align}
    It is easy to check that $R_{ab},\rho_a\geq 0$ and $\tr\rho_a=1$, and that
    $\rho_a\otimes \one_b \pm \TC_b(R_{ab})\geq 0$, which ensures that  the
    pair $(R_{ab},\rho_a)$ is indeed feasible in \eqref{eq:gamma-sdp}.  To
    compute the objective value $\tr R_{ab}J^{s}_{ab}$ with $J^s_{ab}\equiv
    J^{\NC_s}_{ab}$, observe that the 
    $J^s_{ab}$ has the form
    \begin{align}
	        J^s_{ab} &= [\psi_1]_{ab} + [\psi_2]_{ab}c\label{eq:Ns-choi}\\
        |\psi_1\rangle_{ab} &= \sqrt{s} |00\rangle_{ab} + |12\rangle_{ab}\\
        |\psi_2\rangle_{ab} &= \sqrt{1-s} |01\rangle + |22\rangle.
	\end{align}
    We have $\langle \psi_1|R_{ab}|\psi_1\rangle_{ab} = s/2$ and $\langle
    \psi_2|R_{ab}|\psi_2\rangle_{ab} = 1 + \sqrt{1-s}-s/2$, and hence
    \begin{align}
    \Gamma(\NC_s) \geq \tr R_{ab}J^s_{ab} = \langle
        \psi_1|R_{ab}|\psi_1\rangle_{ab} + \langle
        \psi_2|R_{ab}|\psi_2\rangle_{ab} = 1 + \sqrt{1-s}.
	\end{align}

    To prove the third inequality in \eqref{eq:Qcap-bounds},
    $\|\TC\circ\NC_s\|_\diamond \leq 1+\sqrt{1-s}$, we again pick feasible
    operators in the SDP \eqref{eq:diamond-norm-sdp} for
    $\|\TC\circ\NC_s\|_\diamond$:
    \begin{align}
        Y_{ab} = Z_{ab} &= s [00]_{ab} + (1-s) [01]_{ab} + \sqrt{1-s}[02]_{ab}
        + [12]_{ab} + [22]_{ab} + [\phi]_{ab} \,,\\ |\phi\rangle_{ab} &=
        \sqrt[4]{\frac{s^2}{1-s}} |10\rangle_{ab} + \sqrt[4]{1-s}
        |21\rangle_{ab} \,.
	\end{align}
    Evidently, $Y_{ab}\geq 0$.  Furthermore, the operator 
    $\begin{pmatrix} 
        Y_{ab} & -\TC_b(J^s_{ab})\\ 
        -\TC_b(J^s_{ab}) & Y_{ab}
    \end{pmatrix}$
	is unitarily equivalent to the following operator in block-diagonal form:
    \begin{align}
	\begin{pmatrix} 
        Y_{ab} & -\TC_b(J^s_{ab})\\ 
        -\TC_b(J^s_{ab}) & Y_{ab}
        \end{pmatrix}  
        \sim s M 	\oplus 	(1-s) M \oplus M \otimes \one_2 \oplus
        \frac{s}{\sqrt{1-s}} [\phi_1] \oplus \sqrt{1-s} [\phi_2] 
        \label{eq:block-diagonal-op}
	\end{align}
    with the matrix
    $M=\left(\begin{smallmatrix}\phantom{-}1&-1\\-1&\phantom{-}1\end{smallmatrix}\right)$
        and vectors
    \begin{align}
        |\phi_1\rangle &= -\sqrt{\frac{1-s}{s}} |0\rangle + |1\rangle + \sqrt{\frac{1-s}{s}} |2\rangle, \qquad 
        |\phi_2\rangle = \sqrt{\frac{s}{1-s}} |0\rangle + |1\rangle -|2\rangle.
	\end{align}
    As the operator on the right-hand side of \eqref{eq:block-diagonal-op} is
    manifestly positive semidefinite, the same holds for 
    \begin{align}
    \begin{pmatrix}
    Y_{ab} & -\TC_b(J^s_{ab})\\ 
    -\TC_b(J^s_{ab}) & Y_{ab}
    \end{pmatrix},
    \end{align}
    showing that $Y_{ab}$ and $Z_{ab}=Y_{ab}$ are feasible in
    \eqref{eq:diamond-norm-sdp}.  The marginal $Y_a=\tr_b Y_{ab}$ is diagonal
    with eigenvalues $1+\sqrt{1-s}$ (of multiplicity 2) and $1+s/\sqrt{1-s}$.
    Since $\sqrt{1-s}\geq s/\sqrt{1-s}$ for $s\in[0,1/2]$, we conclude that
    $\|Y_{a}\|_\infty = 1+\sqrt{1-s}$.  Therefore, the SDP
    \eqref{eq:diamond-norm-sdp} for $\|\TC\circ\NC_s\|_\diamond$ has value at
    most $\|Y_a\|_\infty = 1+\sqrt{1-s}$, which concludes the proof of the
    theorem.
\end{proof}

While the bound $\log\Gamma(\BC)$ can be strictly tighter than the
transposition bound $\log\|\TC\circ\BC\|_\diamond$ for certain channels $\BC$
\cite{WangFangEA18}, Theorem~\ref{thm:qc-upper-bounds} shows that the two
bounds in fact coincide for $\NC_s$.  The SDP upper bound $\log\Gamma(\BC)$ was
recently improved by Fawzi and Fawzi~\cite{FangFawzi19}, and evaluating the latter bound
(which can again be computed by semidefinite programming) yields an even
tighter bound on $\QC(\NC_s)$.  The two bounds are compared in
Fig.~\ref{fig:qcap-bounds-comparison}. 

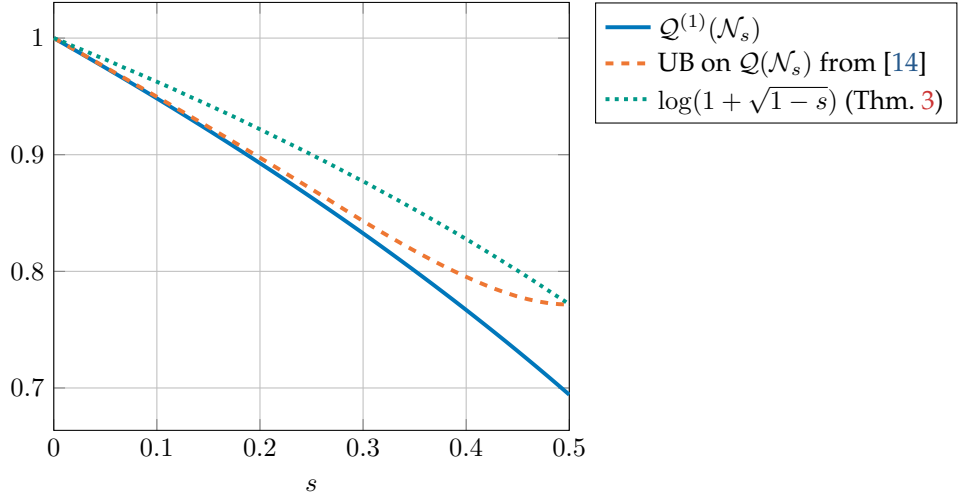
\begin{figure}
	\centering
	\begin{tikzpicture}
	\begin{axis}[
	xlabel=$s$,
	xmin = 0,
	xmax = 0.5,
	scale=1,
	every axis plot/.append style={line width=1.5pt},
	legend cell align={left},
	legend columns = 1,
	legend style={at = {(1.05,1)},anchor = north west,/tikz/every even column/.append style={column sep=.65em}},
	grid = both,
	]
	\addplot[mark=none,color=plotblue] table[x=s,y=ci] {capacities.dat};
	\addplot[mark=none,dashed,color=plotorange] table[x=s,y=ub] {capacities.dat};
	\addplot[dotted,domain=0:1,samples=100,color=plotgreen] {ln(1+(1-x)^(.5))/ln(2)};
	\legend{$\QC^{(1)}(\NC_s)$,UB on $\QC(\NC_s)$ from \cite{FangFawzi19},$\log(1+\sqrt{1-s})$ (Thm.~\ref{thm:Q})};
	\end{axis}
	\end{tikzpicture}
	\caption{Lower and upper bounds on the quantum capacity of $\NC_s$. Plotted are the coherent information $\QC^{(1)}(\NC_s)$ (solid blue), the upper bound (UB) $\hat{R}_\alpha(\NC_s)$ on $\QC(\NC_s)$ from \cite{FangFawzi19} with $\alpha = 1+2^{-5}$ (dashed orange), and the UB from Theorem~\ref{thm:Q} (dotted green).}
	\label{fig:qcap-bounds-comparison}
\end{figure}

\subsection{Upper bound on $\QC(\MC_d)$}

We now derive an analytical upper bound on the quantum capacity of $\MC_d$,
which we introduced in Section~\ref{sec:MdChan} as a generalization of
$\NC_{1/2}$ to arbitrary dimension $d$ of the input and output Hilbert spaces.
As a reminder, a channel isometry for $\MC_d$ is given by $G\colon \HC_a\to
\HC_b\otimes\HC_c$ with $\HC_a\cong\HC_b\cong \mathbb{C}^d$ and
$\HC_c\cong\mathbb{C}^{d-1}$, defined via the following action on the
computational basis $\lbrace |j\rangle_a\rbrace_{j=0}^{d-1}$ of $\HC_a$:
\begin{align}
    \begin{aligned}
    G\colon |0\rangle_a &\longmapsto \frac{1}{\sqrt{d-1}} \sum_{j=0}^{d-2} |j\rangle_b|j\rangle_c\\
    |j\rangle_a &\longmapsto |d-1\rangle_b |j-1\rangle_c\quad\text{for }j=1,\dots,d-1.
    \end{aligned}
    \label{eq:Md}
\end{align}
The Choi-Jamio\l{}kowski operator $J_{ab}^{\MC_d}$ of $\MC_d$ is given as
follows:

\begin{align}
    J_{ab}^{\MC_d} &= \sum_{j=0}^{d-2} \left(\frac{1}{d-1} [0j] + [j+1,d-1] +
    \frac{1}{\sqrt{d-1}}\left(|0,j\rangle\langle j+1,d-1| +
    |j+1,d-1\rangle\langle 0,j|\right)\right),\label{eq:Md-Choi}
\end{align}
where we used the notation $[\psi]\equiv |\psi\rangle\langle\psi|$.

Using Proposition~\ref{thm:qc-upper-bounds} and the SDP
\eqref{eq:diamond-norm-sdp}, we now derive an upper bound on the quantum
capacity of the channel $\MC_d$.

\begin{theorem}
	\label{thm:Q-Md}
	For any $d\geq 2$, 
	\begin{align}
	\QC(\MC_d) \leq \log\left(1+\frac{1}{\sqrt{d-1}}\right).
	\end{align}
	In particular, $\QC(\MC_d)\to 0$ as $d\to\infty$.
\end{theorem}

\begin{proof}
    We will prove this theorem by constructing feasible operators $Y_{ab} =
    Z_{ab}$ in the SDP \eqref{eq:diamond-norm-sdp} for
    $\|\TC\circ\MC_d\|_\diamond$ satisfying $\|\tr_b Y_{ab}\|_\infty = 1 +
    \frac{1}{\sqrt{d-1}}$, from which the claim follows by invoking
    Proposition~\ref{thm:qc-upper-bounds}.
	
    Consider the following ansatz for $Y_{ab}$:
    \begin{align}
    Y_{ab} &= \frac{1}{d-1} \sum_{j=0}^{d-2} [0]\otimes [j] +
        \frac{1}{\sqrt{d-1}} [0] \otimes [d-1] + \sum_{j=1}^{d-1} [j] \otimes
        [d-1] + [\Xi],\label{eq:Y-ansatz}\\ \text{with}\quad |\Xi\rangle &=
        \frac{1}{\sqrt[4]{d-1}}\sum_{j=1}^{d-1} |j\rangle |j-1\rangle.
	\end{align}
    We clearly have $Y_{ab}\geq 0$.  Setting $Z_{ab} = Y_{ab}$, consider the
    second feasibility constraint in \eqref{eq:diamond-norm-sdp},
    \begin{align}
        \begin{pmatrix}
            Y_{ab} & - \TC_b(J_{ab}^{\MC_d})\\ 
            - \TC_b(J_{ab}^{\MC_d}) & Y_{ab}
        \end{pmatrix} \geq 0.
	\end{align}
	Taking the Schur complement, this is equivalent to the constraint
    \begin{align}
	    Y_{ab} \geq \TC_b(J_{ab}^{\MC_d})\, Y_{ab}^{-1}\, \TC_b(J_{ab}^{\MC_d})
        \label{eq:Schur-complement}
	\end{align}
    with the inverse taken on the support of $Y_{ab}$.  Using the form of
    $J_{ab}^{\MC_d}$ in \eqref{eq:Md-Choi}, it is straightforward to check that
    $Y_{ab} = \TC_b(J_{ab}^{\MC_d})\, Y_{ab}^{-1}\, \TC_b(J_{ab}^{\MC_d})$ so
    that \eqref{eq:Schur-complement} is satisfied.  Finally, it follows from
    \eqref{eq:Y-ansatz} that 
    \begin{align}
	    Y_a = \tr_b{Y_{ab}} = \left(1+\frac{1}{\sqrt{d-1}}\right)\one_a,
	\end{align} 
    which yields $\|Y_a\|_\infty = 1+\frac{1}{\sqrt{d-1}}$ and concludes the
    proof.
\end{proof}

\section{Private and classical capacities}\label{sec:pccap}

The private and classical capacities of the channels $\NC_s$ and $\MC_d$ can be
determined exactly.  In the following, we will prove in
Theorems~\ref{thm:PCcap-Ns} and \ref{thm:PCcap-Md} that $\PC(\NC_s) =
\CC(\NC_s) = 1 = \PC(\MC_d) = \CC(\MC_d)$.  This is remarkable because neither
$\NC_s$ nor $\MC_d$ belong to any of the special classes of channels for which
the private or classical capacity is known to have a single-letter expression
(that is, equal to the channel private information or the Holevo information;
see Section~\ref{sec:prelim} for the definitions of these quantities).

The upper bound on the private and classical capacities of the channels employ
the SDP technique used in prior work of Wang, Xie, and
Duan~\cite{WangXieEA17}. That same work applied the technique to a
qutrit-to-qutrit channel unitarily equivalent to $\NC_s$~(see
discussion below eq.~\eqref{isoDef1}).
Hence, these prior upper and lower bounds
on $\CC$~(see Prop.~15 in~\cite{WangXieEA17}) and $\CC_E$~(see Prop.~1
in~\cite{WangDuan18}) imply our bounds on $\CC(\NC_s)$ in
Th.~\ref{thm:PCcap-Ns} and $\CC_E(\NC_s)$ in Th.~\ref{thm:entAssist}. 
On the other hand, our lower bound on $\PC$, first inequality in~\eqref{eq:capacity-chain} below, correctly proves the bound previously stated in \cite[Prop.~16]{WangXieEA17}, whose proof contained a typo. 
The corresponding results for the $\MC_d$ channel do not follow from
these prior works.

\subsection{Capacities of $\NC_s$}
\label{sec:classical-caps-Ns}
We will obtain a lower bound for the channel private information and the Holevo
information using the following equiprobable ensemble of two quantum states:
\begin{align}
    \rho_a^1 = \begin{pmatrix}1 & 0&0\\0&0&0\\0&0&0\end{pmatrix}\,, & 
    \quad
    \rho_a^2 = \begin{pmatrix}0 & 0&0\\0&s&0\\0&0&1-s\end{pmatrix}.
\label{eq:private-ensemble}
\end{align}
For this ensemble, the quantity $\Dl(\BC,\bar{\rho}_a) - \sum\nolimits_x p_x
\Dl(\BC,\rho_a^x)$ evaluates to $1$ which gives a lower bound to the channel
private information $\PC^{(1)}(\NC_s)\geq 1$ for any $s\in[0,1/2]$ (see
\eqref{eq:private-information}).  Likewise, the Holevo information is $1$ for
this ensemble, giving a lower bound $1\leq \chi(\NC_s)$ by \eqref{eq:chi}.  
Using \cite{CaiWinterEA04,Devetak05} and \eqref{eq:PC-relation}, we obtain
chains of inequalities
\begin{align}
    \begin{aligned}
    1 & \leq \PC^{(1)}(\NC_s) \leq \PC(\NC_s) \leq \CC(\NC_s)\,,\\
    1 & \leq \PC^{(1)}(\NC_s) \leq \chi(\NC_s) \leq \CC(\NC_s).
    \label{eq:capacity-chain}
    \end{aligned}
\end{align}
We will now show that $\CC(\NC_s) \leq 1$ so we have equalities throughout the
above.  
To this end, we employ a semidefinite programming upper bound on the classical
capacity derived by Wang et al.~\cite{WangXieEA17}:
\begin{proposition}[\cite{WangXieEA17}]
	\label{prop:beta-bound}
	For any quantum channel $\BC$, 
    \begin{align}
	    \CC(\BC) \leq \log \beta(\BC),
	\end{align} 
	where $\beta(\BC)$ is the solution of the following SDP:
    \begin{align}
    \begin{aligned} 
    \beta(\BC) = \text{min.~} & \tr S_b\\
    \text{s.t.~} & R_{ab}, S_{b} \text{ Hermitian}\\
    & {-}R_{ab} \leq \TC_b(J^\BC_{ab}) \leq R_{ab}\\
    & {-}\one_a \otimes S_b \leq \TC_b(R_{ab}) \leq \one_a\otimes S_b
    \end{aligned}
    \label{eq:beta-sdp}
    \end{align}
    Moreover, $\log \beta(\BC)$ is a strong converse bound: if the classical
    information transmission rate exceeds $\log \beta(\BC)$, the transmission
    error converges to $1$ exponentially fast.
\end{proposition}

We now state and prove the main result of this section:
\begin{theorem}\label{thm:PCcap-Ns}
	For all $s\in[0,1/2]$, 
	\begin{align}
	\PC(\NC_s) = \CC(\NC_s) = 1.
	\end{align}
	Moreover, both the classical and the private capacity of $\NC_s$ satisfy the strong converse property.
\end{theorem}

\begin{proof}
    Due to \eqref{eq:capacity-chain} and Proposition~\ref{prop:beta-bound}, it
    suffices to show that $\beta(\NC_s) \leq 2$ for all $s\in[0,1/2]$.  To this
    end, let
    \begin{align}
        \begin{aligned}
        R_{ab} &= s[00]_{ab} + (1-s)[01]_{ab} + [02]_{ab} + [\psi]_{ab} + [12]_{ab} + [22]_{ab}\\
        S_b &= s[0]_b + (1-s)[1]_b + [2]_b,
        \end{aligned}
	\label{eq:R-S-ansatz-Ns}
	\end{align}
    where $|\psi\rangle_{ab} = \sqrt{s} |10\rangle + \sqrt{1-s} |21\rangle$.
    We now check that $(R_{ab},S_b)$ is feasible for the SDP
    \eqref{eq:beta-sdp}.  Both $R_{ab}$ and $S_b$ are Hermitian by
    construction, and one readily checks that $\one_a\otimes S_b \pm
    \TC_b(R_{ab}) \geq 0$.  Recalling the form of $J^s_{ab}\equiv
    J^{\NC_s}_{ab}$ from \eqref{eq:Ns-choi}, it also follows that $R_{ab} \pm
    \TC_b(J^s_{ab}) \geq 0$.  These observations establish feasibility of
    $(R_{ab},S_b)$ in \eqref{eq:beta-sdp}.  Furthermore $\tr S_b = 2$, hence
    $\beta(\NC_s)\leq 2$.  (As a side remark, we note that $(R_{ab},S_b)$ is in
    fact optimal for \eqref{eq:beta-sdp} since $1\leq \CC(\NC_s)\leq \log
    \beta(\NC_s)\leq 1$ by \eqref{eq:capacity-chain} and
    Proposition~\ref{prop:beta-bound}.)
	
    Since $\log\beta(\NC_s)$ is a strong converse bound and $\CC(\NC_s) =
    \log\beta(\NC_s) = 1$, the classical capacity of $\NC_s$ satisfies the
    strong converse property.	This also holds for the private capacity
    $\PC(\NC_s) = \CC(\NC_s) = 1$, which can be seen as follows (we refer to
    \cite{Wilde16} for precise definitions).  Let $\varepsilon_\CC$ denote the
    error for a classical information transmission code.  The error
    $\varepsilon_\PC$ for a private information transmission code is defined as
    $\varepsilon_\PC = \max(\varepsilon_\CC,\varepsilon_{\mathrm{env}})$, where
    $\varepsilon_{\mathrm{env}}$ is an additional error parameter controlling
    how much information the environment gains about Alice's input.  Assume now
    that we have a private information transmission code with rate $r_\PC > 1$
    and error $\varepsilon_\PC$.  Then this code can also be regarded as a
    classical information transmission code with rate $r_\CC = r_\PC > 1$ and
    error $\varepsilon_\CC \leq \varepsilon_\PC$.  The strong converse property
    of $\CC(\NC_s)$ implies that $\varepsilon_\CC \to 1$ as the code
    blocklength $n$ increases, from which $\varepsilon_\PC\to 1$ follows.
    Hence, $\PC(\NC_s)$ also satisfies the strong converse property, which
    concludes the proof. 
\end{proof}

Theorems~\ref{thm:Q} and \ref{thm:PCcap-Ns} imply that the quantum and private
capacities of $\NC_s$ are strictly separated:

\begin{corollary}\label{cor:QP-separation}
	$\QC(\NC_s)<\PC(\NC_s)$ for all $s\in(0,1/2]$.
\end{corollary}

It was recently shown in \cite{ding2020bounding} that the quantity $\log\beta(\cdot)$ from \Cref{prop:beta-bound} also serves as an upper bound on the classical capacity of a quantum channel assisted by a classical feedback channel, denoted $\CC_{\leftarrow}(\cdot)$.
Hence, \Cref{thm:PCcap-Ns} immediately implies that the feedback-assisted classical capacity of $\NC_s$ is equal to its classical capacity, $\CC(\NC_s)=1=\CC_{\leftarrow}(\NC_s)$ for all $s\in[0,1/2]$.
Moreover, the same argument applies to the private capacity of a quantum channel assisted by a public feedback channel.

Finally, we discuss the entanglement-assisted capacity of the channel $\NC_s$.
Similar to the classical capacity, it is independent of the parameter $s$:
\begin{theorem}~\label{thm:entAssist}
	For all $s\in[0,1/2]$,
	\begin{align}
	\CC_{E}(\NC_s) = 2.
	\end{align}
\end{theorem}

\begin{proof}
   It is easy to check that the state $\sigma_{ab} = (\IC_a \otimes
    \NC_s)([\psi_{aa'}])$ with
    \begin{align}
   |\psi\rangle_{aa'} = \frac{1}{\sqrt{2}} |00\rangle_{aa'} +
        \frac{1}{\sqrt{2}} \left( \sqrt{s} |11\rangle_{aa'} + \sqrt{1-s}
        |22\rangle_{aa'}\right)
   \end{align}
    achieves $I(A;B)_\sigma = 2$.  The optimality of $[\psi_{aa'}]$ can be
    verified using the methods of \cite{FawziFawzi18}.
\end{proof}

\subsection{Capacities of $\MC_d$}\label{sec:pccap-Md}

We proved in Theorem~\ref{thm:PCcap-Ns} in Section~\ref{sec:classical-caps-Ns}
that one use of the channel $\NC_s$ can faithfully transmit 1 private bit (and
thus also 1 classical bit) regardless of the value of $s$.  We now prove that
the $d$-dimensional generalization $\MC_d$ of $\NC_{1/2}$ defined in
\eqref{eq:Md} retains unit private and classical capacity for any local
dimension $d$:

\begin{theorem}\label{thm:PCcap-Md}
	For all $d\geq 2$,
	\begin{align}
	\PC(\MC_d) = \CC(\MC_d) = 1.
	\end{align}
	Moreover, both the classical and private capacity of $\MC_d$ satisfy the strong converse property.
\end{theorem}

\begin{proof}
    The proof strategy is similar to the one used in
    Theorem~\ref{thm:PCcap-Ns}.  First, consider an equiprobable ensemble with
    the following two quantum states,
    \begin{align}
        \rho_a^0 = [0]_a, & \quad \rho_a^1 = \frac{1}{d-1} \sum_{j=1}^{d-1} [j]_a,
        \label{eq:Md-priv-ensemble}
	\end{align}
    and form the cqq state 
    \begin{align}
        \sigma_{\mathsf{x}bc} = (\IC_{\mathsf{x}} \otimes G) \rho_{\mathsf{x}a} (\IC_{\mathsf{x}} \otimes G)^\dagger,
	\end{align}
    where $G\colon \HC_a\to \HC_b\otimes \HC_c$ is the channel isometry for
    $\MC_d$ defined in \eqref{eq:Md}, and $\rho_{\mathsf{x}a} = \frac{1}{2}
    \left( [0]_{\mathsf{x}}\otimes \rho_a^0 + [1]_{\mathsf{x}}\otimes \rho_a^1
    \right)$.  It is straightforward to check that $I(X;B)_\sigma=1$ and
    $I(X;E)_\sigma = 0$, from which we obtain
    \begin{align}
	    1 \leq \PC^{(1)}(\MC_d) \leq \PC(\MC_d) \leq \CC(\MC_d).\label{eq:Md-chain}
	\end{align}
	
    The claim of the theorem now follows by showing that $\CC(\MC_d)\leq 1$ for
    all $d$.  To this end, we once again employ the upper bound
    $\log\beta(\MC_d)$ from Proposition~\ref{prop:beta-bound}.  Consider the
    following Hermitian operators $R_{ab}$ and $S_b$ (analoguous to the
    operators in \eqref{eq:R-S-ansatz-Ns}):
    \begin{align}
	\begin{aligned}
		R_{ab} &= \frac{1}{d-1} \sum_{j=0}^{d-2} [0]_a \otimes [j]_b + \sum_{j=0}^{d-1} [j]_a \otimes [d-1] + [\psi]_{ab} \,, \\
		S_b &= \frac{1}{d-1} \sum_{j=0}^{d-2} [j]_a + [d-1]_a \,,
	\end{aligned}
	\end{align}
    where $|\psi\rangle_{ab} = \frac{1}{\sqrt{d-1}} \sum_{j=1}^{d-1}
    |j\rangle_a \otimes |j-1\rangle_b$.  One readily checks that $R_{ab}$ and
    $S_{b}$ are feasible in the SDP \eqref{eq:beta-sdp}, that is, $R_{ab}\pm
    \TC_b(J_{ab}^{\MC_d}) \geq 0$ and  $\one_a\otimes S_b \pm \TC_b(R_{ab})
    \geq 0$.  Furthermore, $\tr S_b = 2$ for any $d$, and hence $\beta(\MC_d)
    \leq 2$.  Using Proposition~\ref{prop:beta-bound}, we conclude $\CC(\MC_d)
    \leq \log \beta(\MC_d) \leq 1$, which together with \eqref{eq:Md-chain}
    gives
	\begin{align}
	    1 \leq \PC(\MC_d) \leq \CC(\MC_d) \leq 1.
	\end{align}
    The strong converse property for $\CC(\MC_d)$ and $\PC(\MC_d)$ follows in
    the same way as in the proof of Theorem~\ref{thm:PCcap-Ns}.
\end{proof}

By the same argument as in the remark after \Cref{cor:QP-separation}, \Cref{thm:PCcap-Md} implies that the feedback-assisted private and classical capacities of $\MC_d$ are equal to their unassisted counterparts.

\section{Discussion of capacities of the platypus channels}
\label{sec:cap-discussion}

We now summarize the findings of Sections~\ref{sec:CI}, \ref{sec:qcap-bounds}
and \ref{sec:pccap} on the quantum capacity $\QC$, private capacity $\PC$ and
classical capacity $\CC$ of the platypus channels $\NC_s$ (with $s\in(0,1/2]$) and $\MC_d$ (for $d\geq 3$):
%
\begin{align}
	0 < \QC^{(1)}(\NC_s) \stackrel{?}{=} \QC(\NC_s) \leq \log\left(1+\sqrt{1-s}\right) &< 1 = \PC^{(1)}(\NC_s) = \PC(\NC_s)
	= \chi(\NC_s) = \CC(\NC_s), \label{eq:capacities-Ns}\\
	0 < \QC^{(1)}(\MC_d) \stackrel{?}{=} \QC(\MC_d) \leq \log\left(1+\frac{1}{\sqrt{d-1}}\right) &< 1 = \PC^{(1)}(\MC_d) = \PC(\MC_d)
	= \chi(\MC_d) = \CC(\MC_d).  \label{eq:capacities-Md}
\end{align}
In the above \cref{eq:capacities-Ns}, the left-most equality labeled by ``?'' is the conjectured weak
additivity of the single-letter coherent information, $\QC^{(1)}(\NC_s)$, which
would be implied by the validity of the ``spin alignment conjecture'' described
in Section~\ref{sec:spinAl}.  The next inequality is Theorem~\ref{thm:Q}.
Finally, the four equalities on the RHS of \eqref{eq:capacities-Ns} come from
Theorem~\ref{thm:PCcap-Ns} and \eqref{eq:capacity-chain}.
Likewise, in \cref{eq:capacities-Md} the conjectured equality labeled by ``?'' would be implied by the validity of the (higher-dimensional version of the) spin-alignment conjecture in Section~\ref{sec:conjecture-d-dim}, and the following inequality and equalities are obtained via Theorems~\ref{thm:Q-Md} and \ref{thm:PCcap-Md}, respectively.
For both $\NC_s$ and $\MC_d$, the private and classical capacity have the strong converse property,
as proved in Theorems~\ref{thm:PCcap-Ns} and \ref{thm:PCcap-Md}, respectively.  These findings are remarkable for
various reasons:

\begin{itemize}
    \item The private information $\PC^{(1)}(\cdot)$ is additive for both $\NC_s$ and $\MC_d$.  The only
        known classes of quantum channels with additive private information are
        (a) ``less noisy channels'' $\BC$ whose complementary channels have
        vanishing private capacity, $\PC(\BC^c) = 0$ \cite{Watanabe12} and of
        which degradable channels \cite{DevetakShor05,Smith08} are special
        cases; (b) anti-degradable channels; and (c) direct sums of partial
        traces (DSPT), a special case of the ternary ring of operators (TRO)
        channels \cite{GaoJungeEA18}. We know from Section~\ref{sec:channel-defs}
        that $\PC(\NC_s^c) = 1$ for all $s\in[0,1/2]$ and $\PC(\MC_d^c) = \log(d-1)$ for $d\geq 3$.
        Hence, neither $\NC_s$ nor $\MC_d$ are 
        less noisy.  Moreover, clearly neither of these channels is a direct sum of partial traces, so that both $\NC_s$ and $\MC_d$ fall outside all known classes of channels with additive private information.

    \item The Holevo information $\chi(\cdot)$ is additive  for both $\NC_s$ and $\MC_d$.  Again, both channels
        fall outside of all the known classes of channels with additive Holevo
        information:   (a) entanglement-breaking channels \cite{Shor02}; (b)
        unital qubit channels \cite{King02}; (c) depolarizing channels
        \cite{King03}; (d) Hadamard channels \cite{King06,KingMatsumotoEA07};
        (e) DSPT channels \cite{GaoJungeEA18}; and (f) erasure channels
        \cite{BennettDiVincenzoEA97}.  Since entanglement-breaking channels
        have vanishing quantum capacity and both $\NC_s$ and $\MC_d$ have positive quantum
        capacity, the platypus channels are not entanglement-breaking. They clearly do not
        belong to classes (b) and (c) either. A quantum channel is Hadamard if
        its complementary channel is entanglement-breaking \cite{Wilde16}.
        Since $Q(\NC_s^c) = 1$ and $Q(\MC_d^c) = \log(d-1)$, the complements of $\NC_s$ and $\MC_d$ cannot be
        entanglement-breaking, so that neither channel is Hadamard.  Finally, neither $\NC_s$ nor $\MC_d$
	    are a DSPT nor an erasure channel.

    \item The quantum capacity of both $\NC_s$ and $\MC_d$ is strictly smaller than their respective private
        capacities, for all $s\in(0,1/2]$ and $d\geq 3$. There are not
        too many examples of this phenomenon.  The first known class is the
        Horodecki channels, for which the quantum capacity vanishes and the
        private capacity is strictly positive
        \cite{HorodeckiHorodeckiEA05,HorodeckiHorodeckiEA09,OzolsSmithSmolin13}.
        The smallest such example has input and output dimensions $d_a = d_b =
        3, d_c = 4$, and the separation is typically small.  The second class
        is the so-called ``half-rocket channels,'' with quantum capacity
        between $0.6$ and $1$ but private capacity $\log d$ where the input and
        output dimenions are $d_a = d^2, d_b = d_c = d^6-d^4$.  This class
        exhibits an extensive separation of the two
        capacities~\cite{LeungLiEA14}.  In comparison, $\NC_s$ is the smallest
        known channel with $d_a = d_b = 3$, $d_c = 2$ exhibiting the
        separation, and the separation is quite large (at least $\approx s/2$
        for most $s$ of interest).  A separation of the information quantities
        $\QC^{(1)}(\cdot)$ and $\PC^{(1)}(\cdot)$ was observed for certain
        channels for which these quantities are also superadditive, such as the
        depolarizing channel \cite{DiVincenzoShorEA98,SmithRenesEA08} or the
        dephrasure channel \cite{LeditzkyLeungEA18}.  However, due to
        super-additivity in these channels we do not know their exact
        capacities and the true separations between them.
    
    \item Reference~\cite{HorodeckiHorodeckiEA05} shows that a quantum state
        can yield a bit of classical information that is private from the
        environment $E$ if and only if it is of the form
        \begin{align}
            \label{HoroPbit}
            \raisebox{0.5ex}{$\gamma$}_{ {K_A} {K_B} {S_A} {S_B}}  = \, 
            U \left( \, [\phi]_{ {K_A} {K_B} } \otimes \, \raisebox{0.2ex}{$\sigma$}_{ {S_A} {S_B} } \right) U^\dagger
        \end{align}
        where $K_A K_B$ are called the key systems, $S_A S_B$ are called the
        shield systems, $\ket{\phi}$ is the maximally entangled state on $K_A
        K_B$, $\sigma$ is an arbitrary state on $S_A S_B$ and $U$ is a
        controlled unitary of the form $\sum_{i,j} \, [ij]_{K_A K_B} \otimes
        (U_{ij})_{S_A S_B}$, with each $U_{ij}$ a unitary that depends on $i$
        and $j$.  The key is shared between two users Alice and Bob.  Alice is
        in possession of $K_A S_A$, Bob is in possession of $K_B S_B$, and they
        can generate a key by measuring along the computational basis of $K_A
        K_B$ independently.  Furthermore, if the one-way distillable key of the
        state is strictly greater than the one-way distillable entanglement,
        each of the shield systems must be nontrivial.  

        Since $\NC_s$ has 3 dimensional input and output and can send one bit
        privately with a single use, it can be used to make a $3 \times 3$
        dimensional state shared by Alice and Bob that encodes one private bit.
        Furthermore, this state is not of the form Eq.~(\ref{HoroPbit}).  To
        see this, suppose the contrary.  By the quantum capacity bound of
        $\NC_s$, this state has one-way distillable entanglement strictly less
        than $1$, so each of the shield systems must be nontrivial with at
        least $2$ dimensions.  Meanwhile, the key systems have $4$ dimensions
        jointly, so, the total dimension exceeds $9$ which is a contradiction.  

        The resolution is that this state is locally equivalent to a standard
        p-bit with each of $K_A$, $K_B$, $S_A$, $S_B$ being a qubit, but for
        which the local ranks of both $K_A S_A$ and $K_B S_B$ are 3.  So, our
        state \emph{is} a p-bit, but one which has been embedded into smaller
        dimensional spaces than would be possible generically.

        In more detail, here is the protocol to create the $3 \times 3$ state
        which distributes a private key between Alice and Bob.  Alice prepares
        the state
        \begin{align}
            \frac{1}{\sqrt{2}} \; \ket{0}_A\ket{0}_{\tilde{A}}
            + \frac{\sqrt{s}}{\sqrt{2}} \; \ket{1}_A\ket{1}_{\tilde{A}}
            + \frac{\sqrt{1{-}s}}{\sqrt{2}} \; \ket{2}_A\ket{2}_{\tilde{A}}
        \end{align}
        and applies $F_S$ (the isometry giving rise to $\NC_s$) to $\tilde{A}$
        resulting in the state
        \begin{align}
            \ket{\nu}_{A B E} = \frac{1}{\sqrt{2}} \; \ket{0}_A (\sqrt{s}
            \ket{0}_{B} \ket{0}_{E} + \sqrt{1{-}s} \ket{1}_{B} \ket{1}_{E})   +
            \frac{\sqrt{s}}{\sqrt{2}} \; \ket{1}_A \ket{2}_{B} \ket{0}_E +
            \frac{\sqrt{1{-}s}}{\sqrt{2}} \; \ket{2}_A \ket{2}_{B} \ket{1}_E
            \,.
        \label{eq:3by3}
        \end{align}
        The systems $A,B$ are neither key systems nor shield systems.  Consider
        the local isometries:
        \begin{align}
            \begin{array}{lll}
            \ket{0}_A \rightarrow \ket{0}_{K_A} \ket{0}_{S_A}, & & \ket{0}_B \rightarrow \ket{0}_{K_B} \ket{0}_{S_B}, \\
            \ket{1}_A \rightarrow \ket{1}_{K_A} \ket{0}_{S_A}, & & \ket{1}_B \rightarrow \ket{0}_{K_B} \ket{1}_{S_B}, \\
            \ket{2}_A \rightarrow \ket{1}_{K_A} \ket{1}_{S_A}, & & \ket{2}_B \rightarrow \ket{1}_{K_B} \ket{0}_{S_B}.
            \end{array}
        \end{align}
        Applying the above local isometries to $\ket{\nu}_{A B E}$ results in 
        \begin{align}
        \nonumber
              & \frac{1}{\sqrt{2}} \; \ket{0}_{K_A} \ket{0}_{S_A}
              (\sqrt{s} \ket{0}_{K_B} \ket{0}_{S_B} \ket{0}_{E} + \sqrt{1{-}s} \ket{0}_{K_B} \ket{1}_{S_B} \ket{1}_{E})   
        \\    + & \frac{\sqrt{s}}{\sqrt{2}} \; \ket{1}_{K_A} \ket{0}_{S_A} \ket{1}_{K_B} \ket{0}_{S_B}  \ket{0}_E
            + \frac{\sqrt{1{-}s}}{\sqrt{2}} \; \ket{1}_{K_A} \ket{1}_{S_A} \ket{1}_{K_B} \ket{0}_{S_B} \ket{1}_E \,
        \label{eq:4by4}
        \\ = & \frac{1}{\sqrt{2}} \; \ket{0}_{K_A} \ket{0}_{K_B} 
        (\sqrt{s} \ket{0}_{S_A} \ket{0}_{S_B} \ket{0}_{E} + \sqrt{1{-}s} \ket{0}_{S_A} \ket{1}_{S_B} \ket{1}_{E})
        \nonumber 
        \\ + & \frac{1}{\sqrt{2}} \; \ket{1}_{K_A} \ket{1}_{K_B} 
        (\sqrt{s} \ket{0}_{S_A} \ket{0}_{S_B} \ket{0}_{E} + \sqrt{1{-}s} \ket{1}_{S_A} \ket{0}_{S_B} \ket{1}_{E})
        \end{align}
        If we trace out $E$ from the above, we get a state of the form
        Eq.~(\ref{HoroPbit}) where $U_{00} = U_{01} = U_{10} = I$, $U_{11}$ is
        the swap operator, and $\sigma = |00\rangle \langle 00| + |01\rangle
        \langle 01|$.

    \item Both the private and the classical capacity of $\NC_s$ and $\MC_d$ satisfy the strong
        converse property.  For the private capacity, the strong converse
        property is only known for (a) a subclass of degradable channels called
        generalized dephasing channels~\cite{TomamichelWildeEA17} (and for
        these channels, the quantum and private capacities coincide, $\QC =
        \PC$); (b) DSPT channels \cite{GaoJungeEA18}.  Both $\NC_s$ and $\MC_d$ are
        provably non-degradable and not of DSPT form, and hence fall outside
        both classes.  For the classical capacity, the strong converse is known
        for a number of channel classes: (a)~erasure channels
        \cite{WildeWinter14}; (b) depolarizing channels, unital qubit channels,
        and the Holevo-Werner channel \cite{KoenigWehner09}; (c)
        entanglement-breaking and Hadamard channels \cite{WildeWinterEA14}; (d)
        DSPT channels \cite{GaoJungeEA18}. By the arguments made above,
        neither $\NC_s$ nor $\MC_d$ belong to any of these classes.

    \item Finally, the coherent information $\QC^{(1)}(\cdot)$ is additive for both $\NC_s$ and $\MC_d$
        relative to the corresponding version of the spin alignment conjecture. The known classes of
        channels with additive coherent information are (a) less noisy channels
        \cite{Watanabe12} (which includes degradable channels
        \cite{DevetakShor05}); (b) anti-degradable channels; (c) PPT channels
        \cite{HorodeckiHorodeckiEA00}; (d) DSPT channels \cite{GaoJungeEA18}.
        The channel $\NC_s$ is neither degradable nor anti-degradable.  Since
        PPT channels have vanishing quantum capacity, $\NC_s$ and $\MC_d$ cannot be PPT
        either.
        
    \item The $\NC_s$ channel is unitarily equivalent to a qutrit-qutrit channel $\LC_\alpha$ that was introduced in \cite{WangDuan18} to study zero-error capacities. 
    In the follow-up work \cite{WangXieEA17}, the authors showed that the private and classical capacity of $\LC_\alpha$ coincide (see Sec.~\ref{sec:pccap} for a more detailed discussion) and satisfy the strong converse property, also noting that $\LC_{\alpha}$ (and hence also $\NC_s$) does not belong to any of the known classes of channels with that property for the private or classical capacity listed above.
    They further announced (without proof) an analytical upper bound on the quantum capacity of $\LC_\alpha$ separating it from the private capacity.
    This bound coincides with our upper bound, for which we give a full proof in Thm.~\ref{thm:Q}.
    In our independent study we construct the related channel $\NC_s$ as a hybrid of two simple channels (see Sec.~\ref{sec:NsChan}) and analyze in detail the additivity properties of the various information quantities of $\NC_s$. 
    Furthermore, we extend the channel construction to a larger family of channels of arbitrary dimension with similar information-theoretic properties (see Sec.~\ref{sec:MdChan} and \ref{sec:glnChan}).
    In the process, we also give full proofs of some of the statements announced in \cite{WangXieEA17} about the capacities of the unitarily equivalent channel $\LC_\alpha$.
\end{itemize}

\section{A 3-parameter generalization of $F_s$}
\label{sec:3parameter}

We now discuss generalizations of the channel $\NC_s$ and their capacities, to
further our understanding of the phenomena exhibited by $\NC_s$.  The
one-parameter isometry $F_s:\HC_a \mapsto \HC_b \ot \HC_c$ in~\eqref{isoDef1}
for $0 \leq s \leq 1$ has input dimension $d_a=3$, output dimension $d_b = 3$
and environment dimension $d_c=2$.  The isometry $F_s$ can be generalized by adding two additional
parameters, $\mu, \nu \in [0,1]$ without changing the input, output, and
environment spaces, leading to an isometry $V_{s,\mu,\nu}: \HC_a \mapsto \HC_b
\ot \HC_c$ that acts as
\begin{align}
\begin{aligned}
    V_{s,\mu,\nu} \, \ket{0} &= \sqrt{s}  \, \ket{0}_b\ot \ket{0}_c + \sqrt{1-s}   \; \ket{1}_b \ot \ket{1}_c \,,\\
    V_{s,\mu,\nu} \, \ket{1} &= \sqrt{\nu} \, \ket{1}_b\ot \ket{0}_c + \sqrt{1-\nu} \; \ket{2}_b \ot \ket{1}_c \,,\\
    V_{s,\mu,\nu} \, \ket{2} &= \sqrt{\mu} \, \ket{2}_b\ot \ket{0}_c + \sqrt{1-\mu} \; \ket{0}_b \ot \ket{1}_c \,.\\
\end{aligned}
\label{3PrmIso}
\end{align}
We use ${\cal W}_{s,\mu,\nu}$ to denote the resulting channel from $\hat \HC_a$
to $\hat \HC_b$.  
The isometry $V_{s,\mu,\nu}$ becomes $F_s$~\eqref{isoDef1} when $\nu=\mu=1$, so
$V_{s,\mu,\nu}$ and ${\cal W}_{s,\mu,\nu}$ indeed generalize $F_s$ and $\NC_s$
respectively.  

We study the degradability of the channel ${\cal W}_{s,\mu,\nu}$ using the
framework of \cite{SiddhuGriffiths16}.  We call a channel {\em pcubed} if it is
generated by a pcubed isometry.  We call an isometry pcubed if there exists a
basis of the input space that is mapped by the isometry to product states of
the output space and the environment space. This special basis of the input
space is not required to be orthogonal.   
It is straightforward to verify that, when $0 <s, \mu,\nu < 1 $, the isometry
$V_{s,\mu,\nu}$ is a {\em pcubed} isometry:   
\begin{align}
    V_{s,\mu,\nu} \, \ket{\al_i} = \ket{\bt_i} \ot \ket{\gm_i} \quad \text{for} \; 0 \leq i \leq 2 \,,
    \label{eq:p3prm}
\end{align}
where 
\begin{align}
   \begin{aligned}
    \ket{\al_i} &= a_0 \big( \ket{0} + \om^i k_1 \ket{1} + \om^{2i} k_2 \ket{2} \big), \\
    \ket{\bt_i} &= b_0 \big( \ket{0} + \om^i l_1 \ket{1} + \om^{2i} l_2 \ket{2} \big),\\
    \ket{\gm_i} &= c_0 \big( \ket{0} + \om^{2i} r\ket{2} \big) \,, 
   \end{aligned}
\end{align}
$a_0,b_0,$ and $c_0$ are constants that respectively normalize $\ket{\al_i},
\ket{\bt_i}, \ket{\gm_i}$, $\om$ is a cube root of unity, and $k_1,k_2,l_1,l_2$
and $r$ are non-negative numbers related to $s, \mu, \nu$ as follows:
\begin{align}
    \begin{aligned}
    k_1 &= \frac{1}{r} \sqrt{ \frac{1-s}{\nu}}\,, & \qquad\quad l_1 &= \frac{1}{r} \sqrt{ \frac{1-s}{s}} \,, & \qquad\quad r &= \left(\frac{(1-s)(1-\mu)(1-\nu)}{s\mu\nu}\right)^{1/6}\hspace*{-3ex},\\
    k_2 &= r \sqrt{ \frac{s}{1-\mu}} \,,& l_2 &= r \sqrt{ \frac{\mu}{1-\mu}}\,.
    \end{aligned}
\end{align}

The inner products between the states witnessing the pcubed channel also
characterize its degradability.  
To this end, let $A$, $B$, $C$ be the Gram matrices for the sets
$\{\ket{\al_i}\}, \{\ket{\bt_i}\}$, and $\{ \ket{\gm_i} \}$, respectively: 
\begin{align}
    A_{jk} = \inpd{\al_j}{\al_k}\,, \quad B_{jk} = \inpd{\bt_j}{\bt_k}\,,  
    \quad C_{jk} = \inpd{\gm_j}{\gm_k} \,.  
\end{align}
Each of $A,B,C$ has the form 
\begin{align}
    M = 
    \begin{pmatrix}
        1  & m & m^* \\
        m^*& 1 & m \\
        m & m^*& 1
    \end{pmatrix},
    \label{grmM}
\end{align}
where $m^*$ is the complex conjugate of $m$, and $M = A, B, C$ respectively
when $m$ is set to be 
\begin{align}
    a = |a_0|^2(1 + \om k_1^2 + \om^2 k_2^2), \quad
    b = |b_0|^2(1 + \om l_1^2 + \om^2 l_2^2), \quad \text{and} \quad
    c = |c_0|^2(1 + \om^2 r^2) \,,
    \label{eq:abcVals}
\end{align}
respectively.  As a side remark, since $V_{s, \mu,\nu}$ is an isometry, it
follows that $A = B*C$ where $*$ denotes the elementwise or Hadamard Product of
two matrices.

The channel ${\cal W}_{s, \mu, \nu}$ is degradable if and only if there is a
Gram matrix $D$ satisfying
\begin{align}
    B = C * D.
    \label{eq:gramDeg}
\end{align}
To see this, when such a Gram matrix $D$ exists, there are normalized kets
$\{\ket{\dl_i}\}$ in some auxiliary Hilbert space $\HC_d$ such that $D_{jk} =
\inpd{\dl_j}{\dl_k}$.
A possible degrading map can be generated by the pcubed isometry from $\HC_b
\mapsto \HC_c \ot \HC_d$ taking $\ket{\bt_i}$ to $\ket{\gm_i} \otimes
\ket{\dl_i}$.  The converse follows from \eqref{eq:p3prm} since the degradable
map must take $\ket{\bt_i}$ to $\ket{\gm_i}$, and must be generated by an
isometry.   (See Sec.~III.C in~\cite{SiddhuGriffiths16} for a detailed
discussion.)

\subsection{The isometry $V_{s,\mu,1-\mu}$}\label{sec:2parameter}

We now consider a two parameter subclass of isometries, $W_{s,\mu}$, obtained
from setting $\nu = 1-\mu$ in $V_{s,\mu,\nu}$, with $s\in[0,1/2]$ and
$\mu\in[0,1]$.  Following \eqref{3PrmIso}, 
\begin{align} 
\begin{aligned}
    W_{s,\mu} \, |0\rangle_a &= \sqrt{s}\, \ket{0}_b \ot \ket{0}_c + \sqrt{1-s}\,\ket{1}_b \ot \ket{1}_c \,, \\
    W_{s,\mu} \, |1\rangle_a &= \sqrt{1-\mu} \,\ket{1}_b \ot \ket{0}_c + \sqrt{\mu} \, \ket{2}_b \ot \ket{1}_c| \,, \\
    W_{s,\mu} \, |2\rangle_a &= \sqrt{\mu} \, \ket{2}_b \ot \ket{0}_c + \sqrt{1-\mu} \, \ket{0}_b \ot \ket{1}_c \,.
\end{aligned}
\label{intIso}
\end{align}
The resulting channel $\WC_{s,\mu}=\tr_c(W_{s,\mu}\cdot W_{s,\mu}^\dagger)$ has
two Kraus operators
\begin{align}
    K_0 &= \begin{pmatrix}
    \sqrt{s} & 0 & 0\\
    0 & \sqrt{1-\mu} & 0\\
    0 & 0 & \sqrt{\mu} 
    \end{pmatrix} \,,
    & 
    K_1 &= \begin{pmatrix}
    0 & 0 &\sqrt{1-\mu}\\
    \sqrt{1-s} & 0 & 0\\
    0 &\sqrt{\mu}& 0 
    \end{pmatrix} \,.
\end{align}
As a side remark, when $\mu=1$, the channel $\WC_{s,1}$ is unitarily equivalent
to $\NC_s$: we have $\NC_s = \WC_{s,1}(U\cdot U^\dagger)$, where the unitary
$U$ swaps $|1\rangle_a$ and $|2\rangle_a$ at the input.

For the rest of the discussion we focus on $s=1/2$.  We will first evaluate the
capacities of $\WC_{1/2,1/2}$.  Then, we study the degradability of
$\WC_{1/2,\mu}$, followed by a detailed numerical analysis of its capacities.

When $\mu=1/2$, $K_0 = \frac{1}{\sqrt{2}}\one$, and $K_1$ is proportional to
the qutrit-$X$ Heisenberg-Weyl operator, which acts as $X |i\rangle =
|i+1\!\!\mod 3\rangle$.
Hence, $\WC_{1/2,1/2}$ is the qutrit $X$-dephasing channel with dephasing
probability $1/2$, which fixes the eigenbasis of $X$:
\begin{align}
    \begin{aligned} 
    |\psi_0\rangle_a &= \frac{1}{\sqrt{3}} (|0\rangle_a + |1\rangle_a + |2\rangle_a)\\
    |\psi_1\rangle_a &= \frac{1}{\sqrt{3}} (|0\rangle_a + \omega |1\rangle_a + \omega^2 |2\rangle_a)\\
    |\psi_2\rangle_a &= \frac{1}{\sqrt{3}} (|0\rangle_a + \omega^2 |1\rangle_a + \omega |2\rangle_a) \,.
    \end{aligned}
    \label{eq:dephasing-basis}
\end{align}
We can now evaluate the capacities.  The above invariant basis can be used to
transmit $\log d$ bits, and this code saturates the dimension bound for the
classical capacity.  Thus, the Holevo information is additive and
$\chi(\WC_{1/2,1/2}) = \CC(\WC_{1/2,1/2}) = \log d$.  For the quantum capacity,
one may evaluate the SDP upper bound given in \eqref{eq:gamma-sdp} for
$\WC_{1/2,1/2}$ giving a value of $\log(3/2)$.  This upper bound is in turn
achieved by the input state $\rho_a = \frac{1}{3}\one$ giving coherent
information $\Delta(\WC_{1/2,1/2},\rho_a) = \log(3/2)$.  Furthermore,
$\WC_{1/2,1/2}$ is degradable\footnote{The channel $\WC_{1/2,1/2}$ can be
generated by the isometry which attaches $\ket{+}_c$ to system $a$, and
conditioned on system $c$ being in the state $\ket{1}$, the unitary $K_1$ is
applied to system $a$ which is then relabeled as system $b$.  This isometry
applied to system $b$ generates a valid degrading map.} so the private and
quantum capacities coincide.    Finally, the channel mutual information
$I(\WC_{1/2,1/2})$ is equal to $\log(9/2)$, since the input $\rho_a =
\frac{1}{3}\one$ simultaneously maximizes the first term and separately the
second and third terms combined in \eqref{eq:entasscap}. 

Altogether, for the quantum and private capacities we have

\begin{subequations}
	\begin{align}
	    \QC(\WC_{1/2,1/2}) &= \QC^{(1)}(\WC_{1/2,1/2}) = \PC^{(1)}(\WC_{1/2,1/2}) = \PC(\WC_{1/2,1/2}) = \log\frac{3}{2},
    \intertext{while for the classical capacities we have}
    \CC(\WC_{1/2,1/2}) &= \chi(\WC_{1/2,1/2}) = \log 3 \,, \\
    \CC_{E}(\WC_{1/2,1/2}) &= \log \frac{9}{2} \,.
\end{align}
    \label{eq:caps-dephasing}
\end{subequations}

Keeping $s=1/2$, we now return to general $\mu \in [0,1]$ in studying
$\WC_{1/2,\mu}$. 
We first use the pcubed framework to study degradability of this
subclass of channels.  The values of $a,b,$ and $c$ in~\eqref{eq:abcVals} are
given by 
\begin{align}
    a &= \frac{1-2\mu}{2(2-\mu)}\,, & b &= \om^2\frac{2\mu-1}{2-\mu}\,, & c &= -\om/2 \,.
\end{align}
Setting $m$ in \eqref{grmM} to $a,b,$ and $c$ gives the Gram matrices $A,B,$
and $C$ respectively.  The channel $\WC_{1/2,\mu}$ is degradable iff $B=C*D$
for some Gram matrix $D$ (see~\eqref{eq:gramDeg}).  If such a matrix exists, it
has the form given by $M$ in \eqref{grmM} with $m = d = 2\om (1- 2\mu)/(2 -
\mu)$, and thus is a valid Gram matrix (being positive semi-definite) when $1 +
2 d^3 - 3 |d|^2 \geq 0$ (see eq.~(57) in~\cite{SiddhuGriffiths16}).  Setting $d
= 2\om (1 - 2\mu)/(2 - \mu)$ gives $\mu \leq 2/3$.    Thus, $\WC_{1/2,\mu}$ is
degradable iff $\mu \leq 2/3$.

We now evaluate lower and upper bounds on the quantum, private, and classical
capacities of the channel $\WC_{1/2,\mu}$ for $\mu\in[0,1]$, collected in
Fig.~\ref{fig:wsm}.  
The bounds on the quantum capacity $\QC(\WC_{1/2,\mu})$ are obtained by
numerically optimizing the single-letter coherent information
$\QC^{(1)}(\WC_{1/2,\mu})$ (solid blue line in Fig.~\ref{fig:wsm}) and the SDP
upper bound \eqref{eq:gamma-sdp} (dashed blue line in Fig.~\ref{fig:wsm}).
Interestingly, the (possibly tighter) SDP upper bound from \cite{FangFawzi19}
coincides with the SDP bound \eqref{eq:gamma-sdp} from \cite{WangFangEA18} for
all  $\mu\in[0,1]$.  To bound the classical capacity $\CC(\WC_{1/2,\mu})$, we
numerically optimize the single-letter Holevo information $\chi(\WC_{1/2,\mu})$
(solid green line in Fig.~\ref{fig:wsm}) and evaluate the SDP upper bound
\eqref{eq:beta-sdp} (dashed green line in Fig.~\ref{fig:wsm}).  The
entanglement-assisted capacity $\CC_{E}(\WC_{1/2,\mu})$ is computed using the
technique developed in \cite{FawziFawzi18} (solid magenta line in
Fig.~\ref{fig:wsm}).  Finally, for the private capacity $\PC(\WC_{1/2,\mu})$ we
numerically optimize the single-letter private information
$\PC^{(1)}(\WC_{1/2,\mu})$ (solid orange line in Fig.~\ref{fig:wsm}).  To
obtain an upper bound on the private capacity we employ the following recent
result by Fawzi and Fawzi~\cite{FangFawzi19} providing a bound on the private capacity of a
quantum channel in terms of a conic program:
\begin{proposition}[{\cite{FangFawzi19}}]
	\label{prop:P-ub-FF}
    Let $\BC\colon \hat{\HC}_a\to \hat{\HC}_b$ be a quantum channel with
    (unnormalized) Choi-Jamio\l{}kowski operator $J_{ab}^\BC$.  Let furthermore
    $l\in\mathbb{N}$ and set $\alpha = 1+2^{-l}$.  Then we have 
    \begin{align}
	    \PC(\BC) \leq \hat{E}_{\alpha}(\BC),
	\end{align}
    where $\hat{E}_{\alpha}(\BC) = l 2^l -(2^l+1)\log (2^l+1) + (2^l+1)\log
    T_\alpha(\BC)$, and $T_\alpha(\BC)$ is the solution of the following conic
    program:
    \begin{align}
        \begin{aligned} 
            T_\alpha(\BC) = \text{\normalfont max.~} & \tr \left[J_{ab}^\BC \left( K\herm - \sum\nolimits_{i=1}^l W_i\right)\right]\\
            \text{\normalfont s.t.~} & \lbrace W_i\rbrace_{i=1}^l \in\Herm(\HC_a\otimes\HC_b)\\
            & K,\lbrace Z_i\rbrace_{i=0}^l \in \LC(\HC_{a}\otimes \HC_b)\\
            & \rho_a\in \Herm(\HC_a), \tr\rho_a = 1\\
        &\begin{pmatrix}
        \rho_a \otimes \one_b & K\\ K^\dagger &Z_l\herm 
        \end{pmatrix} \geq 0\\
        & \begin{pmatrix}
        W_i & Z_i\\ Z_i^\dagger & Z_{i-1}\herm 
        \end{pmatrix}\geq 0 \quad\text{\normalfont for all $i=1,\dots,l$}\\
        & \rho_a \otimes \one_b - Z_0\herm \text{\normalfont ~block-positive.}
        \end{aligned}
    \label{eq:pc-conic}
	\end{align}
\end{proposition}

In the above, $X\herm = X+X^\dagger$, and an operator $X_{ab}$ is
block-positive (with respect to the bipartition $a : b$) if
$(\langle\psi_a|\otimes \langle \phi_b|) X_{ab} (|\psi_a\rangle \otimes |
\phi_b\rangle) \geq 0$ for all $|\psi_a\rangle\in\HC_a$ and
$|\phi_b\rangle\in\HC_b$.  Hence, block-positive bipartite states are the 
states Choi-Jamio\l{}kowski operators of positive maps (see, e.g.,
\cite{Wolf12}).

The conic program in Proposition~\ref{prop:P-ub-FF} only reduces to an SDP if
$d_ad_b \leq 6$, whereas our channel $\WC_{1/2,\mu}$ has qutrit input and
output, $d_a=d_b=3$.  However, the following strategy suggested to us by Hamza
Fawzi \cite{Fawzi20} may be employed to obtain an (SDP-computable) upper bound
on the quantity $T_\alpha(\BC)$ (and hence $\hat{E}_\alpha(\BC)$), in turn
giving an upper bound on the private capacity of a channel $\BC$ via
Proposition~\ref{prop:P-ub-FF}:

\begin{lemma}[{\cite{FangFawzi19,Fawzi20}}]
	\label{lem:P-ub}
    With the same notation as in Proposition~\ref{prop:P-ub-FF}, we have the
    following bound on the private capacity of a quantum channel $\BC\colon
    \hat{\HC}_a\to \hat{\HC}_b$:
    \begin{align}
        \PC(\BC) \leq \min_{\lbrace \phi_a^i\rbrace_{i=1}^N} F_\alpha(\BC,\lbrace \phi_a^i\rbrace_{i=1}^N),
    \end{align}
    where $N\in\mathbb{N}$ is some fixed natural number, the minimization is over
    sets of pure states $\phi_a^i\in\HC_a$, $i=1,\dots,N$, and the quantity
    $F_\alpha(\BC,\lbrace \phi_a^i\rbrace_{i=1}^N)$ is defined as
    \begin{align}
    F_{\alpha}(\BC) = l 2^l -(2^l+1)\log (2^l+1) + (2^l+1)\log U_\alpha(\BC,\lbrace \phi_a^i\rbrace_{i=1}^N).
    \end{align}
    In the above, $U_\alpha(\BC,\lbrace \phi_a^i\rbrace_{i=1}^N)$ is the solution
    of the following semidefinite program:
    \begin{align}
    \begin{aligned} 
    U_\alpha(\BC,\lbrace \phi_a^i\rbrace_{i=1}^N) = \text{\normalfont max.~} & \tr \left[J_{ab}^\BC \left( K\herm - \sum\nolimits_{i=1}^l W_i\right)\right]\\
    \text{\normalfont s.t.~} & \lbrace W_i\rbrace_{i=1}^l \in\Herm(\HC_a\otimes\HC_b)\\
    & K,\lbrace Z_i\rbrace_{i=0}^l \in \LC(\HC_{a}\otimes \HC_b)\\
    & \rho_a\in \Herm(\HC_a), \tr\rho_a = 1\\
    &\begin{pmatrix}
    \rho_a \otimes \one_b & K\\ K^\dagger &Z_l\herm 
    \end{pmatrix} \geq 0\\
    & \begin{pmatrix}
    W_i & Z_i\\ Z_i^\dagger & Z_{i-1}\herm 
    \end{pmatrix}\geq 0 \quad\text{\normalfont for all $i=1,\dots,l$}\\
    & \tr_a \left( \sigma_{ab} \left( (\phi_a^i)^T \otimes \one_b\right) \right) \geq 0 \quad\text{for $i=1,\dots,N$,}
    \end{aligned}
    \label{eq:pc-sdp}
    \end{align}
    where $\sigma_{ab} = \rho_a \otimes \one_b - Z_0\herm$.
\end{lemma}

\begin{proof}
    The block positivity constraint on $\sigma_{ab} = \rho_a \otimes \one_b -
    Z_0\herm$ translates via the Choi isomorphism to positivity of the map
    $\Psi\colon \hat{\HC}_a\to \hat{\HC}_b$ whose Choi-Jamio\l{}kowski operator
    is $\sigma_{ab}$, i.e., $\Psi(\chi_a)\geq 0$ for all pure state
    $|\chi\rangle_a\in\HC_a$.  Relaxing this positivity constraint to only
    requiring $\Psi(\phi_a^i)\geq 0$ for some fixed pure states $\phi_a^i$,
    $i=1,\dots,N$ now yields a maximization over a larger set of
    Choi-Jamio\l{}kowski operators resp.~maps $\Psi$, and hence we obtain
    $U_\alpha(\BC,\lbrace \phi_a^i\rbrace_{i=1}^N)\geq T_\alpha(\BC)$, where
    $T_\alpha(\BC)$ is defined in Proposition~\ref{prop:P-ub-FF}.
    Strengthening this bound by minimizing over sets of pure states $\lbrace
    \phi_a^i\rbrace_{i=1}^N$ for fixed $N\in\mathbb{N}$ finishes the proof.
\end{proof}

Any choice of $N\in\mathbb{N}$ and pure states $\lbrace\phi_a^i\rbrace_{i=1}^N$
yields a feasible point in the minimization in Lemma~\ref{lem:P-ub}, and hence
an upper bound on the private capacity of $\BC$.  The states
$\lbrace\phi_a^i\rbrace_{i=1}^N$ can for example be sampled from the Haar
measure.  For our purposes, choosing (a union of) mutually unbiased bases
yields a tighter upper bound.  More precisely, we choose the computational
basis $\lbrace |i\rangle_a\rbrace_{i=0}^2$ and the $X$-eigenbasis $\lbrace
|\psi_i\rangle_a\rbrace_{i=0}^2$, defined in \eqref{eq:dephasing-basis}.  The
resulting bound on $\PC(\WC_{1/2,\mu})$ is plotted in Fig.~\ref{fig:wsm}
(dotted orange line).  We also performed a similar numerical analysis for the
capacities of the complementary channel $\WC^c_{s,\mu}$ in
Figure~\ref{fig:wsm-comp}.

Figure~\ref{fig:wsm} reveals a number of interesting properties of the
1-parameter channel family $\WC_{1/2,\mu}$:
\begin{itemize}
    \item For $\mu\lesssim 0.8$, the coherent information (solid blue line in
        Fig.~\ref{fig:wsm}) and private information (solid orange line in
        Fig.~\ref{fig:wsm}) coincide.  The channel is degradable for
        $\mu\lesssim 0.65$ and anti-degradable for $\mu=0$, see
        Fig.~\ref{fig:deg-antideg}.
	
    \item For $\mu \gtrsim 0.8$, the private information (solid orange line in
        Fig.~\ref{fig:wsm}) is strictly larger than the coherent information
        (solid blue line in Fig.~\ref{fig:wsm}).  For $\mu \gtrsim 0.9$ the
        private information exceeds the SDP upper bound on the quantum capacity
        (dashed blue line in Fig.~\ref{fig:wsm}), hence giving a provable
        separation between quantum and private capacity.
	
    \item The upper bound on $\PC(\WC_{1/2,\mu})$ derived via
        Lemma~\ref{lem:P-ub} (dotted orange line in Fig.~\ref{fig:wsm}) clearly
        separates the private capacity from the classical capacity for all
        $\mu<1$.
	
    \item At $\mu\approx 0.8$ the Holevo information (solid green line in
        Fig.~\ref{fig:wsm}) has an inflection point, changing from concave to
        convex.  For $\mu \lesssim 0.8$ the optimal Holevo information is
        achieved by an ensemble of three pure states, whereas for $\mu\gtrsim
        0.8$ four pure states are needed.  This may be a signature of
        super-additivity of Holevo information, and will be further
        investigated in future work.
\end{itemize}

\begin{figure}
	\centering
	\begin{tikzpicture}
	\begin{axis}[
	xlabel=$\mu$,
	xmin = 0,
	xmax = 1,
	ymin = 0,
	ymax = 2.2,
	scale=1.6,
	every axis plot/.append style={line width=1.5pt},
	legend cell align={left},
	legend columns = 2,
	transpose legend,
	legend style={at = {(0.5,1.075)},anchor = south,/tikz/every even column/.append style={column sep=.65em}},
	grid = both,
	extra y ticks={0.585,1.585,2.17},
	extra y tick labels={$\log\dfrac{3}{2}$,$\log 3$,$\log\dfrac{9}{2}$},
	extra y tick style={ticklabel pos=right, grid style={thick,dashed,color=plotgray}}
	]
	
	\addplot[mark=none,color=plotorange] table[x=mu,y=pi] {cap_bounds.dat};
	\addplot[mark=none,color=plotorange,dotted] table[x=mu,y=pc] {cap_bounds.dat};
	\addplot[mark=none,color=plotblue] table[x=mu,y=ci] {cap_bounds.dat};
	\addplot[mark=none,color=plotblue,dashed] table[x=mu,y=qc] {cap_bounds.dat};
	\addplot[mark=none,color=plotgreen] table[x=mu,y=hi] {cap_bounds.dat};
	\addplot[mark=none,color=plotgreen,dashdotted] table[x=mu,y=cc] {cap_bounds.dat};
	\addplot[mark=none,color=plotmagenta] table[x=mu,y=cea] {cap_bounds.dat};
	\legend{$\PC^{(1)}(\WC_{1/2,\mu})$,UB on $\PC(\WC_{1/2,\mu})$,$\QC^{(1)}(\WC_{1/2,\mu})$,UB on $\QC(\WC_{1/2,\mu})$,$\chi(\WC_{1/2,\mu})$,UB on $\CC(\WC_{1/2,\mu})$,$\CC_{E}(\WC_{1/2,\mu})$};
	\end{axis}
	\end{tikzpicture}
    \caption{ Lower and upper bounds (UB) on the capacities of the quantum
    channel $\WC_{1/2,\mu}$ defined via the isometry \eqref{intIso}.  The
    quantum capacity $\QC(\WC_{1/2,\mu})$ is bounded from below by the
    single-letter coherent information $\QC^{(1)}(\WC_{1/2,\mu})$ (solid blue)
    and from above by the SDP bound \eqref{eq:gamma-sdp} derived in
    \cite{WangFangEA18} (dashed blue).  The private capacity
    $\PC(\WC_{1/2,\mu})$ is bounded from below by the private information
    $\PC^{(1)}(\WC_{1/2,\mu})$ (solid orange) and from above by the bound given
    by Lemma~\ref{lem:P-ub} (dotted orange).  The classical capacity
    $\CC(\WC_{1/2,\mu})$ is bounded from below by the Holevo information
    $\chi(\WC_{1/2,\mu})$ (solid green) and from above by the SDP bound
    \eqref{eq:beta-sdp} derived in \cite{WangXieEA17} (dash-dotted green).  We
    also plot the entanglement-assisted classical capacity
    $\CC_{E}(\WC_{1/2,\mu})$ (solid magenta), computed using the technique in
    \cite{FawziFawzi18}.  For $\mu=\frac{1}{2}$ the channel $\WC_{1/2,1/2}$ is
    a dephasing channel, for which the special values of the capacities from
    \eqref{eq:caps-dephasing} are marked on the right-hand side.  }
	\label{fig:wsm}
\end{figure}

\begin{figure}
	\centering
	\begin{tikzpicture}
	\begin{axis}[
	xlabel=$\mu$,
	xmin = 0,
	xmax = 1,
	ymin = -0.05,
	ymax = 2.05,
	scale=1.6,
	every axis plot/.append style={line width=1.5pt},
	legend cell align={left},
	legend columns = 2,
	transpose legend,
	legend style={at = {(0.5,1.075)},anchor = south,/tikz/every even column/.append style={column sep=.65em}},
	grid = both,
	extra y ticks={1.585},
	extra y tick labels={$\log 3$},
	extra y tick style={ticklabel pos=right, grid style={thick,dashed,color=plotgray}}
	]
	
	\addplot[mark=none,color=plotorange] table[x=mu,y=pi] {comp_cap_bounds.dat};
	\addplot[mark=none,color=plotorange,dotted] table[x=mu,y=pc] {comp_cap_bounds.dat};
	\addplot[mark=none,color=plotblue] table[x=mu,y=ci] {comp_cap_bounds.dat};
	\addplot[mark=none,color=plotblue,dashed] table[x=mu,y=qc] {comp_cap_bounds.dat};
	\addplot[mark=none,color=plotgreen] table[x=mu,y=hi] {comp_cap_bounds.dat};
	\addplot[mark=none,color=plotmagenta] table[x=mu,y=cea] {comp_cap_bounds.dat};
	\legend{$\PC^{(1)}(\WC_{1/2,\mu})$,UB on $\PC(\WC_{1/2,\mu})$,$\QC^{(1)}(\WC_{1/2,\mu})$,UB on $\QC(\WC_{1/2,\mu})$,$\CC(\WC_{1/2,\mu})$,$\CC_{E}(\WC_{1/2,\mu})$};
	\end{axis}
	\end{tikzpicture}
    \caption{ Lower and upper bounds (UB) on the capacities of the
    complementary channel $\WC^c_{1/2,\mu}$ defined via the isometry
    \eqref{intIso}.  The quantum capacity $\QC(\WC^c_{1/2,\mu})$ is bounded
    from below by the single-letter coherent information
    $\QC^{(1)}(\WC^c_{1/2,\mu})$ (solid blue) and from above by the SDP bound
    \eqref{eq:gamma-sdp} derived in \cite{WangFangEA18} (dashed blue).  The
    private capacity $\PC(\WC^c_{1/2,\mu})$ is bounded from below by the
    private information $\PC^{(1)}(\WC^c_{1/2,\mu})$ (solid orange) and from
    above by the bound given by Lemma~\ref{lem:P-ub} (dotted orange).  The
    Holevo information $\chi(\WC^c_{1/2,\mu})$ coincides with the SDP bound
    \eqref{eq:beta-sdp} derived in \cite{WangXieEA17}, and is hence
    (numerically) equal to the classical capacity $\CC(\WC^c_{1/2,\mu})$ (solid
    green).  We also plot the entanglement-assisted classical capacity
    $\CC_{E}(\WC_{1/2,\mu})$ (solid magenta), computed using the technique in
    \cite{FawziFawzi18}.  }
	\label{fig:wsm-comp}
\end{figure}
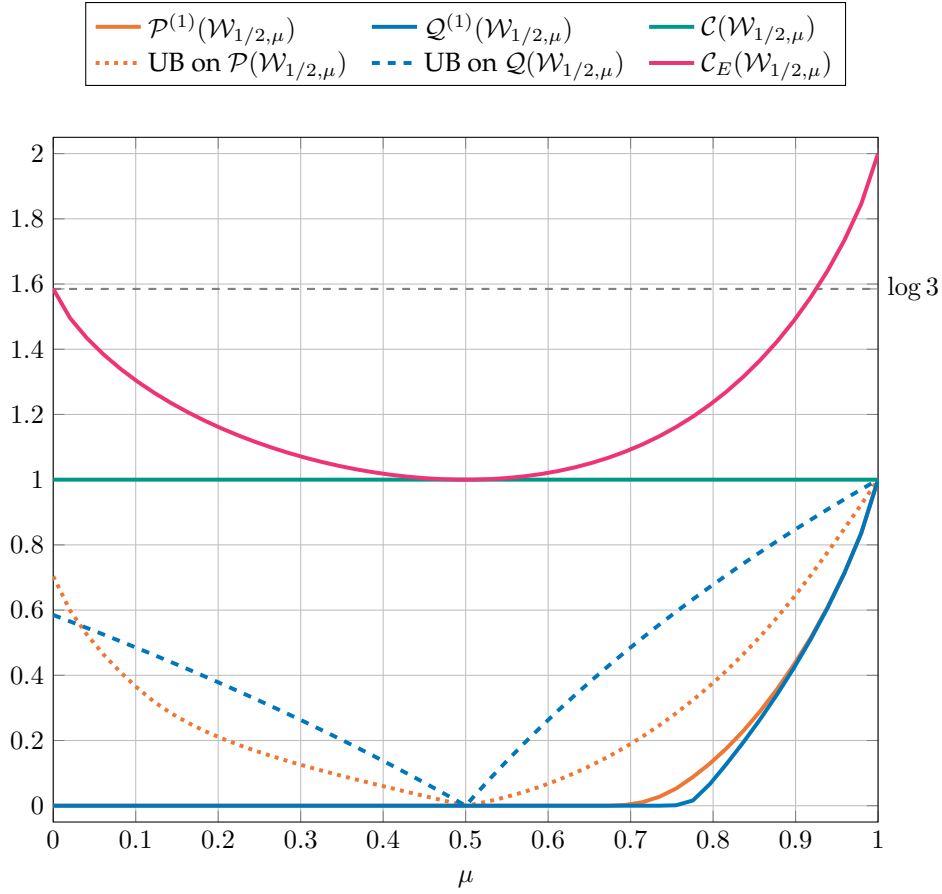

\begin{figure}
	\centering
	\begin{tikzpicture}
	\begin{axis}[
	xlabel=$\mu$,
	xmin = 0,
	xmax = 1,
	scale=1,
	every axis plot/.append style={line width=1.5pt},
	legend cell align={left},
	legend columns = 2,
	legend style={at = {(0.5,1.075)},anchor = south,/tikz/every even column/.append style={column sep=.65em}},
	grid = both,
	ytick distance = 0.25,
	]
	
	\addplot[mark=none,color=plotblue] table[x=mu,y=dg] {degantideg.dat};
	\addplot[mark=none,color=plotmagenta] table[x=mu,y=adg] {degantideg.dat};
	\legend{$\operatorname{dg}(\WC_{1/2,\mu})$,$\operatorname{adg}(\WC_{1/2,\mu})$};
	\end{axis}
	\end{tikzpicture}
    \caption{ Degradability parameter $\operatorname{dg}(\WC_{1/2,\mu})$ (blue)
    and antidegradability parameter $\operatorname{adg}(\WC_{1/2,\mu})$
    (magenta), which are computed using the SDPs from \cite{SutterScholzEA17}.
    A quantum channel $\BC$ is degradable iff $\operatorname{dg}(\BC)=0$, and
    anti-degradable iff $\operatorname{adg}(\BC)=0$.  }
	\label{fig:deg-antideg}
\end{figure}


\section{Concluding remarks}
\label{sec:discussion}

We have studied families of channels that are very simple, yet still
nontrivial in terms of their capacities.  Our goal is to better
understand the boundary between trivially solvable and
incomprehensibly complex behavior.  Our main results demonstrate how
intricately narrow this boundary can be---with complex quantum
effects arising in seemingly innocent and generic settings.
We conclude this paper by highlighting some of these results.

The primary example of quantum channel in this study is obtained by
combining two very simple channels.
As mentioned at the end of Sec.~\ref{sec:NsChan}, we construct the
channel of interest $\NC_s$ by ``hybridizing'' a degradable channel
$\NC_2$ and a completely useless channel $\NC_1$.
The quantum, private, and classical capacities of $\NC_1$ are all
zero.
Meanwhile, the coherent information and various capacities of the
degradable channel $\NC_2$ can be evaluated: $\QC^{(1)}(\NC_2) =
\QC(\NC_2) = \PC(\NC_2) < 1 =  \CC(\NC_2)$.
Both channels have 2-dimensional inputs, and nothing extraordinary
on their own.
We ``stitch'' these channels together, via a common joint output
state to the receiver and the environment.  
The resulting hybrid channel $\NC_s$ has capacities bearing interesting
relationships with those of $\NC_2$: their coherent informations are identical,
$\QC^{(1)}(\NC_s) = \QC^{(1)}(\NC_2)$, and conditioned on the spin alignment
conjecture, $\QC(\NC_s) = \QC(\NC_2)$~(see
Fig.~\ref{fig:qcap-bounds-comparison}).  Meanwhile, both the private and
classical capacity of $\NC_s$ are equal to the classical capacity of $\NC_2$.
In other words, starting from $\NC_2$ and ``stitching'' onto it a completely
useless channel $\NC_1$ boosts the private capacity of $\NC_2$ from its quantum
capacity to its strictly larger classical capacity, while all other quantities
remain the same.

Both the classical and private capacities of $\NC_s$ are, quite intuitively,
equal to 1.  This value as an upper bound comes from the uselessness of the
input state $\ket{2}$ for sending classical information in addition to what can
already be sent using the states $\ket{0}$ and $\ket{1}$, a property inherited
from the uselessness of $\NC_1$.  The remaining input space is $2$-dimensional
so the classical capacity cannot exceed 1.
This value as a lower bound comes from a simple, single-letter, perfect,
private code using two signalling states, $\ket{0}$, along with a mixture of
$\ket{1}$ and $\ket{2}$.  These give rise to respective orthogonal states to
the output but identical states to the environment.
Note that this private code is made possible by the stitching
of $\NC_1$ to $\NC_2$; in fact, this quantum ``stitch'' contributes
to a minimal shield in the p-bit framework discussed in
Sec.~\ref{sec:cap-discussion}.  
Since the private classical rate cannot exceed the classical capacity,
both private and classical capacities must be $1$.
It is highly non-trivial to evaluate these capacities rigorously;
$\NC_s$ does not belong to any known class of channels with additive
Holevo and private informations.
These capacity calculations also prove that that channel's Holevo and
private information are additive in the sense $\chi(\NC_s)
= \CC(\NC_s)$ and $\PC^{(1)}(\NC_s) =
\PC(\NC_s)$.  This additivity however is shown using arguments that
are very different from those used in prior works.

Furthermore, in Sec.~\ref{sec:pccap}, we not only find the classical
and private capacity of $\NC_s$ and its higher dimensional analogue
$\MC_d$, we also provide a strong converse bound for each of these
capacities.  The bound shows that classical or private transmission
rates exceeding the capacity attacts an error that converges to 1
exponentially.  Such strong converse bounds are unavailable for most
quantum channels even when one can compute their capacities.

The quantum capacity of $\NC_s$ is somewhat more complicated, but apparently
can also be understood.  In particular, by restricting to the input space of
$\NC_2$, the quantum capacities are related as $\QC(\NC_2) \leq \QC(\NC_s)$.
Relative to the spin alignment conjecture, we find that $\QC(\NC_s) =
\QC(\NC_2)$ where $\QC(\NC_2)=\QC^{(1)}(\NC_2)$ because of its degradability.
The capacity of $\NC_s$ can be understood this way even though it is neither
degradable nor antidegradable.
A rigorous proof that the capacity of $\NC_s$ is additive would follow from
the spin-alignment conjecture. Such a proof would be qualitatively different
from prior additivity proofs. Our hope is that this spin-alignment conjecture,
which is at heart about the geometric structure of states minimizing
entropy, will lead to further progress on additivity questions in information
theory.

We finally have a family of channels that are not of the usual tractable types
and yet for which we know the classical, private, and quantum capacities.
Underlying this superficial simplicity, the private capacity is much higher
than the quantum capacity, a signature for novel quantum effects at play.  
These channels, and their higher dimensional generalizations, continue to
surprise.  
One may have thought that the weak additivities observed here would
portend strong additivity.  But nothing could be further from the
truth.  In a companion paper \cite{paper2PRL}, we find that the coherent
information of the channel $\NC_s$ tensored with an assisting channel
is super-additive, for a large swath of values of $s$ and for some generically
chosen assisting channel.
The super-additivity can be lifted to quantum capacity for degradable
assisting channels and if the spin alignment conjecture holds.
The assisting channel can have positive or vanishing quantum capacity.
The mechanism behind this superadditivity is novel and in particular
differs from the known explanation of super-activation \cite{SmithYard08,Oppenheim08}.
Additional super-additivity of quantum capacity that is unconditional
on the spin alignment conjecture can be proved for the $d$-dimensional
generalization $\MC_d$ of $\NC_{1/2}$, when it is used with a
$(d{-}1)$-dimensional erasure channel for all nontrivial values of the erasure
probability!
In contrast, the pan-additivity of our channels is a fascinating and
amazing progress.

\paragraph*{Acknowledgments}
We thank Mark M.~Wilde for helpful comments regarding feedback-assisted capacities.
This work was partially supported by ARO MURI Quantum Network Science under
contract number W911NF2120214, NSF grants CCF 1652560, PHY 1915407, and 2137953, and an NSERC discovery grant.

\appendix

\section{Representations of the platypus channel family}

\subsection{$\NC_s$ channel}
Let $\HC_a\cong\HC_b\cong \mathbb{C}^3$ and $\HC_c\cong\mathbb{C}^{2}$, and $s\in[0,1/2]$.
The platypus channel $\NC_s\colon \hat{\HC}_a\to \hat{\HC}_b$ can be defined as follows:

\paragraph{Isometry} $\NC_s(X_a) = \tr_c F_sX_a F_s^\dagger$ with $F_s\colon \HC_a \mapsto \HC_b \ot \HC_c$ defined as
	\begin{align}
		\begin{aligned}
			F_s\colon \ket{0}_a &\longmapsto \sqrt{s} \; \ket{0}_b \ket{0}_c + \sqrt{1-s} \; \ket{1}_b  \ket{1}_c\\
		\ket{1}_a & \longmapsto \ket{2}_b \ket{0}_c\\
		\ket{2}_a & \longmapsto \ket{2}_b \ket{1}_c.
		\end{aligned}
		\label{isoDef1-app}
	\end{align}

\paragraph{Choi operator} $J^{\NC_s}_{ab} = d_a(\IC_a\otimes \NC_s)([\phi])\colon \HC_a\otimes \HC_b \to \HC_a\otimes \HC_b$ ($0$'s are replaced by $.$'s for readability),

\begin{align}
	J^s_{ab} &= \begin{pmatrix}
	s        & .            & . & . & . & \sqrt{s} & . & . & .\\
	.        & 1-s          & . & . & . & .        & . & . & \sqrt{1-s}\\
	.        & .            & . & . & . & .        & . & . & .\\
	.        & .            & . & . & . & .        & . & . & .\\
	.        & .            & . & . & . & .        & . & . & .\\
	\sqrt{s} & .            & . & . & . & 1        & . & . & .\\
	.        & .            & . & . & . & .        & . & . & .\\
	.        & .            & . & . & . & .        & . & . & .\\
	.        & \sqrt{1-s}   & . & . & . & .        & . & . & 1\\
	\end{pmatrix}
	\label{eq:Ns-choi-app}
\end{align}

\paragraph{Kraus operators} $\NC_s(X_a) = N_0 X_a N_0^\dagger + N_1 X_a N_1^\dagger$ with $N_i\colon \HC_a\to \HC_b$ defined as
\begin{align}
	N_0 &= \begin{pmatrix}
		\sqrt{s} & 0 & 0\\ 0 & 0 & 0\\ 0 & 1 & 0
	\end{pmatrix} & 
	N_1 &= \begin{pmatrix}
		0 & 0 & 0\\ \sqrt{1-s} & 0 & 0\\ 0 & 0 & 1
	\end{pmatrix}
\end{align}

\subsection{$\MC_d$ channel}
For $d\in\mathbb{N}$, $d\geq 3$ let $\HC_a\cong\HC_b\cong \mathbb{C}^d$ and $\HC_c\cong\mathbb{C}^{d-1}$.
The platypus channel $\MC_d\colon \hat{\HC}_a\to \hat{\HC}_b$ can be defined as follows:

\noindent\paragraph{Isometry} $\MC_d(X_a) = \tr_c GX_a G^\dagger$ with $G\colon \HC_a\to \HC_b\otimes \HC_c$ defined as
\begin{align}
	\begin{aligned}
		G\colon |0\rangle_a &\longmapsto \frac{1}{\sqrt{d-1}} \sum_{j=0}^{d-2} |j\rangle_b|j\rangle_c\\
		|j\rangle_a &\longmapsto |d-1\rangle_b |j-1\rangle_c\quad\text{for }j=1,\dots,d-1.
	\end{aligned}
	\label{eq:Md-app}
\end{align}

\paragraph{Choi operator} $J_{ab}^{\MC_d} = d_a(\IC_a\otimes \MC_d)([\phi])\colon \HC_a\otimes \HC_b \to \HC_a \otimes \HC_b$,

\begin{multline}
	J_{ab}^{\MC_d} = \sum_{j=0}^{d-2} \left(\frac{1}{d-1} [0j] + [j+1,d-1] \right.\\ 
	{}+ \left.	\frac{1}{\sqrt{d-1}}\left(|0,j\rangle\langle j+1,d-1| +
	|j+1,d-1\rangle\langle 0,j|\right)\right),\label{eq:Md-Choi-app}
\end{multline}
where $[\psi] \equiv |\psi\rangle\langle\psi|$.

\paragraph{Kraus operators} $\MC_d(X_a) = \sum_{k=0}^{d-2} M_k X_a M_k^\dagger$ with $M_k\colon \HC_a\to\HC_b$ for $k=0,\dots d-2$ defined as

\begin{align}
	M_k &= \frac{1}{\sqrt{d-1}} |k\rangle_b \langle 0|_a + |d-1\rangle_b \langle k+1|_a.
\end{align}

\subsection{$\OC$ channel}
For $d\in\mathbb{N}$, $d\geq 2$ let $\HC_a\cong\HC_b\cong \mathbb{C}^d$ and $\HC_c\cong\mathbb{C}^{d-1}$.
For $0\leq j\leq d-2$ let $\mu_j\in\mathbb{C}$ with $\sum_{j=0}^{d-2} |\mu_j|^2 = 1$.
The platypus channel $\OC\colon \hat{\HC}_a\to \hat{\HC}_b$ can be defined as follows:

\paragraph{Isometry} $\OC(X_a) = \tr_c HX_a H^\dagger$ with $H\colon \HC_a \to \HC_b\otimes \HC_c$ defined as
\begin{align}
	\begin{aligned}
		H\colon \ket{0}_a &\longmapsto \sum_{j=0}^{d-2} \mu_j \ket{j}_b\ket{j}_c\\
		\ket{j}_a &\longmapsto \ket{d-1}_b \ket{j-1}_c\quad\text{for }j=1,\dots,d-1.
	\end{aligned}
	\label{isoDef3-app}
\end{align}

\paragraph{Kraus operators} $\OC(X_a) = \sum_{k=0}^{d-2} O_k X_a O_k^\dagger$ with $O_k\colon HC_a\to\HC_b$ for $k=0,\dots d-2$ defined as

\begin{align}
	O_k &= \mu_k |k\rangle_b \langle 0|_a + |d-1\rangle_b \langle k+1|_a.
\end{align}
                                           
\printbibliography[heading=bibintoc]

\end{document}